\documentclass[leqno]{article}
\usepackage{natbib}
\usepackage{amsmath}
\usepackage{amsfonts}
\usepackage{breqn}
\usepackage{nicefrac}
\usepackage[toc,page]{appendix}
\usepackage{color}
\usepackage{tikz}
\usetikzlibrary{decorations.pathmorphing}
\usepackage[utf8]{inputenc}
\usepackage[utf8]{inputenc}
\usepackage[english]{babel}
\usepackage{mathtools}
\usepackage{bbm}
 \usepackage{lipsum}
 \newcommand\numberthis{\addtocounter{equation}{1}\tag{\theequation}}
 
 \newcommand\blfootnote[1]{%
   \begingroup
   \renewcommand\thefootnote{}\footnote{#1}%
   \addtocounter{footnote}{-1}%
   \endgroup
 }

\usepackage[english]{babel}
\usepackage{amsthm}
\usepackage{graphicx}
\usepackage{hyperref}
\newtheorem{theorem}{Theorem}[section]
\newtheorem{lemma}{Lemma}[section]
\newtheorem{remark}{Remark}[section]
\newtheorem{example}{Example}[section]
\newtheorem{definition}{Definition}[section]

\numberwithin{equation}{section}

\usepackage{etoolbox}

\makeatletter

\newif\if@seccntdot

\pretocmd{\section}{\@seccntdottrue}{}{}
\pretocmd{\subsection}{\@seccntdottrue}{}{}
\apptocmd{\@xsect}{\@seccntdotfalse}{}{}

\def\@seccntformat#1{%
  \csname the#1\endcsname
  \if@seccntdot .\fi
  \quad
}

\makeatother

\begin{document}
\renewcommand{\abstractname}{\vspace{-\baselineskip}}

\title{OPTION PRICING WITH DELAYED INFORMATION}

\author{
\vspace{4 mm} \\
Tomoyuki Ichiba \\
 {\normalsize University of California, Santa Barbara}\\
  {\em ichiba@pstat.ucsb.edu} \\
 \vspace{3 mm} \\
 Seyyed Mostafa Mousavi  \\
  {\normalsize University of California, Santa Barbara}\\
  {\em mousavi@pstat.ucsb.edu}}

\date{}

\maketitle
\bigskip

 \begin{abstract} 
 \blfootnote{The authors would like to thank Jean-Pierre Fouque, Mike Ludkovski, Yuri Saporito and Andrey Sarantsev for several helpful discussions and feedbacks at different stages of the work. The first author is supported in part by NSF grants DMS 1313373 and 1615229.}
 We propose a model to study the effects of delayed information on option pricing. We first talk about the absence of arbitrage in our model, and then discuss super replication with delayed information in a binomial model, notably, we present a closed form formula for the price of convex contingent claims. Also, we address the convergence problem as the time-step and delay length tend to zero and introduce analogous results in the continuous time framework. Finally, we explore how delayed information exaggerates the volatility smile.  \vspace{1mm} 
 \end{abstract}
 
\section{Introduction}
All participants in financial markets have access only to delayed information. Delay adds more uncertainty to the market, and it is of great importance to study it. A universal assumption in options pricing literature is that a trader makes his decisions with full access to the prices of the assets (i.e, no delayed information). However, in practice, there is a lag between when the order is decided and its execution time. In particular, there are two important types of delays in financial markets. First is the delay in order execution, that is, the order would be executed with some delay after the trader places it. For example, if the order is made in the morning, it would be executed in the afternoon. Second is the delay in receiving information,  that is, the trader observes the prices and other important information with some delay, usually because of the technological barriers, exacerbated by having long physical distance from the exchange.\par
In the view of traders, these two types of delays act similarly. In both cases, orders are executed with prices which are unknown at the time they are made. In other words, the source of the delayed information does not change the decisions of the trader. For example, let $\{0, 1,\dots \}$ be a discrete trading horizon. If there exists a delay with length of $1$ period, then regardless of what the source of delay is, no trade happens at time $0$, and in later times  trades happen based on the information available up until the previous period. The reason is that if the delay is only in receiving information, then, at time $0$ the trader does not have any information, so he waits till time $1$ to get time-$0$ prices to make a trade and those trades would of course be executed with time-$1$ prices. If the delay is only in order execution, then at time-$0$ and based on time-$0$ prices, the trader makes an order, but that order
 would be executed with time-$1$ prices. \par
 
 In this work, we start with the binomial model proposed by \cite{cox1979option} and consider fixed periods of delay in the flow of information. Therefore, agents have an information stream smaller than the information flow of the traded asset. We show that the market with delayed information is incomplete, and it is not possible to perfectly replicate most contingent claims. Incomplete markets pose various challenges and for a review of different approaches, we refer to \cite{staum2007incomplete}. We take the worst case scenario approach, that is super replication, to price and replicate convex contingent claims. This approach is first suggested by \cite{el1995dynamic} in their seminal paper. We derive recursive and closed-form formulas for pricing convex contingent claims in the discrete time model. Later, we study the continuous time limit as the time-step and delay length tend to zero. We show that the price process under our pricing measure converges to the Black-Scholes price process, but with enlarged volatility.\par 
 
 A very interesting aspect of our model is the way it shows how delayed information affects the volatility smile. Our model confirms the intuition of traders that delayed information would exaggerate the volatility smile, but it does not cause it. We show that in the continuous limit, volatility is constant and there is no smile, but in the discrete model, we can observe volatility smile. In other words, it suggests the idea that the smile observed in the market might not all be by the market itself, and it could have been exaggerated because of the way we interact with delayed information.\par 
 
 Our model with delayed information has some similarities with the models with transaction costs, notably in both models, we encounter similar limit theorems and both risky asset price processes converge to the Black-Scholes price process with some enlarged volatility. In other words, enlarging volatility can be considered as the way to take into account both transaction costs and delayed information. \cite{leland1985option} is first to discuss transaction costs in option pricing models. \cite{boyle1992option} studies transaction costs in binomial models, and \cite{kusuoka1995limit} provides rigorous limit theorems for such models. Some recent works in this area are \cite{bank2016super}, \cite{bank2015scaling} and \cite{dolinsky2016convex} For extensive literature on option pricing with transaction costs, we refer to \cite{kabanov2009markets}.  \par 
 
\cite{kabanov2006dalang} provides an absence of arbitrage condition in discrete time models with delayed information. \cite{kardaras2013generalized} studies market viability in scenarios that the agent has delayed or limited information. Also, \cite{bouchard2015arbitrage}, \cite{burzoni2016pointwise} and \cite{burzoni2014universal} are some very related works in discrete time arbitrage theory.  \par 
 In the literature, in markets with delayed information, risk-minimizing hedging strategies, which is another hedging approach in incomplete markets, have been studied. Using this approach, \cite{di1995hedging} models lack of information by letting the assets to be observed only at discrete times, and \cite{schweizer1994risk} presents the general case of restricted information. Some other works in this direction are \cite{frey2000risk}, \cite{mania2008mean}, \cite{kohlmann2007mean} and \cite{ceci2017follmer}. \par 

The paper is organized as follows. In section \ref{discretemodel}, we set up the discrete time model with delayed information and define the super-replication price. We discuss the super-replicating strategy in an $N$-period binomial model with $H=N-1$ periods of delay in subsection \ref{sec1}, and we generalize the results to an $N$-period binomial model with $H$ periods of delay in subsection \ref{sec2} using both dynamic programming and direct approaches. A geometrical representation of the strategy is presented in subsection \ref{geometry}. In section \ref{continuous}, we study the asymptotic behavior of the model as the time step and delay length tend to zero. In particular, subsection \ref{smile} is devoted to the discussion of how delayed information affects the volatility smile.

 \section{Discrete Time Model}
 \label{discretemodel}
 Before introducing delays, let us recall the $\, N$-period binomial tree model of \cite{cox1979option} for a financial market with a single risky asset and a single risk-free asset (e.g., stock). Given $\, N \in \mathbb N\, $, let us denote by $\, (\Omega,\mathcal{F},\mathbb{P})\, $ a probability space for the canonical space $\, \Omega\, :=\, \{0,1\}^{N}\, $ of $\,N$-period binomial tree with the Borel $\, \sigma$-algebra $\, \mathcal{F}\, $ generated by $\, \Omega\, $. For every $\, \omega \, :=\, 	(\omega_{1},\dots,\omega_{N})\in \Omega\, $ we define a coordinate map by $\, Z_{k}(\omega)\, =\, \omega_{k}\, $ for each  $k \, =\, 1,\dots,N $. Let $\, \mathbb{P}\, $ be the probability measure under which $\, Z_{k}\, $, $\, k \, =\,  1, \ldots , N$ are independent, Bernoulli random variables with $\, \mathbb{P}(Z_{k}=1)\, =\, \mathbb{P}(Z_{k}=0)\, =\, 1\, / \, 2\, $, $\, k \, =\,  1,\dots,N \, $. Define the filtration $ \, \mathfrak F \, :=\,  \{\mathcal{F}_{k}\,,  k = 0, \ldots , N\}\, $, where $\, \mathcal F_{k}\, $ is the  $\, \sigma$-field $\, \sigma(Z_{1},\dots,Z_{k})\, $ generated by the first $\, k\, $ variables for $\, k=1,\dots,N\, $ and $\, \mathcal{F}_{0}\, $ is the trivial $\sigma$-field, i.e., $\mathcal{F}_{0}=\{\emptyset,\Omega\}$. \par
 
 In the $\, N$-period binomial tree model, the risky asset price $\, S_{k}: \Omega \to \mathbb R \, $ and its discounted price $\widetilde{S}_{k}:\Omega \rightarrow \mathbb{R}$, discounted by instantaneous rate $ r > 0$, at time $\, k\, $, are defined by 
 \begin{eqnarray}
 \label{priceprocess}
 S_{k} (\omega) \, :=\, S_{0}\, u^{I_{k}(\omega)}\, d^{k-I_{k}(\omega)} \, , \quad I_{k} (\omega) \, :=\, \sum_{l=1}^{k}Z_{l}(\omega) \, , \quad \widetilde{S}_{k} (\omega) \, :=\, e^{-rk}\, S_k (\omega)\, , \quad k \, =\,  1, \ldots , N \, , 
 \end{eqnarray}
 where $\, S_{0}\, $ is a given initial price of risky asset at time $\, 0\, $, and $\, u\, $ (or $\, d\, $) is a fixed ratio by which the price process goes up (or down) in one period with $\, u > 1+r > d > 0 \,$. The price processes are adapted to the filtration $\,\mathfrak F\,$. 
 

 \subsection{Delayed Filtration}
 We shall introduce delays in the flow of information in the $\, N$-period binomial model. For simplicity, let us consider the situation where an investor sends buy or sell orders to the market at time $\,t\,$, but her orders are not executed until time $\,t + H\,$ with $\, H \in \{0,\dots,N-1\}\, $ {\it delay periods}. The investor herself knows that she has $H$ delay periods when she is sending orders. Then we define the {\it delayed filtration} $\, \mathfrak{G} \, :=\, \{\mathcal G_{k}, k \, =\, 0, 1, \ldots , N\} \, $, where $\,\mathcal {G}_{k} \, := \, \mathcal{F}_{0}\, $, for $\, k \, =\, 0\, , \, \dots\, , H-1\,$, and 
 \begin{eqnarray}
 \mathcal {G}_{k} \, := \, \mathcal{F}_{k-H}\, , \quad k \, =\,  H\, , \, \dots\, , N\, . \quad 
 \end{eqnarray}
 In other words, $\, \mathcal {G}_{k} \, $ is the information set of the price process until time $\, \min ( k-H, 0) \, $, rather than time $\, k\, $.
 In the following, we shall consider investments based on this delayed information.  \par
 
 Let $\, \mathcal{A}_{\mathfrak{G}}\, $ be the set of all $\, \mathfrak{G}\, $-adapted stochastic processes $\, \Delta\,  :=\, \{\Delta_{k}\, , \, k \, =\,  0, \ldots , N-1\}\, $ with $\, \Delta_{k}\, \equiv \, 0\, $, $\, k \, =\, 0, \, \dots\, ,\, H-1\, $. Here, each $\, \Delta \in \mathcal{A}_{\mathfrak{G}}\, $ represents a strategy for this investor based on the delayed information, that is, the positive $\, \Delta_k > 0\, $ (the negative $\, \Delta_{k} < 0\, $, respectively) corresponds to the total number of shares of the risky asset that the investor decides to own (to owe, respectively) at time $\, k\,$, given information $\, \mathcal{G}_{k}\,$. In other words, the order made at time $\, k-H \,$ to buy or sell $\, (\Delta_k - \Delta_{k-1})$ shares of the risky asset, gets executed at time $\, k \,$  with price $\, S_{k}\, $ (not $\, S_{k-H}\, $), because of $\, H\, $ periods of delay. Thus the investor has to deal with the risk of price changes between the time of order submission and execution. 
 
 For an initial investment of $\, x_{0}\, $ in the risk free asset and a strategy $\, \Delta \in \mathcal{A}_{\mathfrak{G}} \, $, we shall consider the portfolio value process $\, V_{k}(x_{0},\Delta) (\omega) \, $, $\, k \, =\, 0,\dots,N\, $, $\,\omega \in \Omega\,$. The first order $\, \Delta_H\, $ submitted at time $\, 0\, $ is executed at time $\, H\, $, and  the portfolio value process is not observed until time $\, H\, $. Thus we define
 \begin{equation}
 \label{initialportfolio}
 V_{H}(x_{0}, \Delta)(\omega) \, :=\,  x_{0} \cdot e^{rH} + \Delta_{H} \cdot {S}_{H}(\omega) \, , \quad V_{0} (x_{0}, \Delta)(\omega)  \, :=\,  e^{-rH} \cdot V_{H}(x_{0},\Delta)(\omega)  \, =\, x_{0} + \Delta_{H} \cdot \widetilde{S}_{H}(\omega) \, , 
 \end{equation}
 and in general 
 \begin{equation}
 \label{portfolio} 
 V_{k}(x_{0},\Delta) (\omega) \, := \, \left \{ \begin{array}{l} 
 e^{-r(H-k)}\cdot V_{H}(x_{0},\Delta)\left(\omega\right), \hspace{3 mm} k=0,\dots,H-1\, , \\ \\
 \, e^{rk} x_{0}+\sum\limits_{l=H}^{k-1}S_{l}(\omega)  \cdot \left(\Delta_{\left(l-1\right) \vee H }-\Delta_{l}\right)+S_{k} (\omega)\cdot \Delta_{(k-1) \vee H}, \hspace{3 mm} k=H,\dots,N \, . 
 \end{array} \right . 
 \end{equation}
 For $\, k \, =\, H,\dots,N\, $, the first term in the portfolio value process ($e^{rk} x_{0}$) in (\ref{portfolio}) corresponds to the initial investment in the risk free asset. The second term ($\sum_{l=H}^{k-1}S_{l}(\omega)  \cdot (\Delta_{\left(l-1\right) \vee H }-\Delta_{l})$) is due to the cash flow in the risk free asset up until time $k$, and the third term ($S_{k} (\omega)\cdot \Delta_{(k-1) \vee H}$) relates to the investment in the risky asset at time $k$. We call $V_k(x_0,\Delta)$, $k=0,\dots,N$ the value process from the strategy $(x_0,\Delta) \in (\mathbb R, \mathcal{A}_{\mathfrak{G}})$.\par 
 
 By construction, the changes in the portfolio value process ($V_{k}(x_{0},\Delta)$) in (\ref{portfolio}) starting from its first realization at time $H$, are only due to the variation in asset prices. In other words, no money is added to or withdrawn from the portfolio.\par 
 
 Note that the initial portfolio value $V_{0} (x_{0}, \Delta)(\omega)$ in (\ref{initialportfolio}) is a random variable, not a constant. This is because it is defined by discounting the time-$H$ portfolio value $V_{H}(x_{0}, \Delta)(\omega)$, which is the first time the portfolio value is observed due to the existence of delay.\par 
 
 For $\, k \, =\, H,\dots,N\, $, $\, \Delta_{k} \, $ is $\,  \mathcal{G}_{k}\, $-measurable, but $\, V_{k}(x_{0},\Delta) \, $ is $\, \mathcal{F}_{k}$-measurable. Thus $V_{k}(x_{0},\Delta)$ is $\, \mathcal{F}_{k \vee H}$-measurable for $\,k \, =\, 0, \ldots , N\,$.  In this sense, the portfolio is constructed based on the delayed information.

 \subsection{Absence of Arbitrage}
 We shall first introduce the notion of arbitrage in our model. In general, arbitrage means that one cannot reap any benefit for free, that is without taking any risk. In our model with delayed information, as it is shown in (\ref{initialportfolio}), the initial portfolio value $V_{0}(x_0,\Delta)$ is a random variable, because of the existence of delay. Therefore, we need to adjust the classical notion of arbitrage in the domain of $(\mathbb{R},\mathcal{A}_{\mathfrak{G}})$ strategies, to take this into account. 
 \theoremstyle{definition}
 \begin{definition}[Arbitrage] 
 \label{arbitrage}
 An arbitrage opportunity is the strategy $(x_0,\Delta) \in (\mathbb{R},\mathcal{A}_{\mathfrak{G}})$ such that 
 \begin{eqnarray}
 \max\limits_{\omega \in \Omega}\{V_{0}(x_{0},\Delta)\left(\omega\right)\} &=& 0, \nonumber \\
 \mathbb{P}(V_{N}(x_{0},\Delta)\geq 0) &=& 1, \\
 \mathbb{P}(V_{N}(x_{0},\Delta) > 0) &>& 0. \nonumber
  \end{eqnarray}
 \end{definition}
 The primary difference with the classical definition of arbitrage is the condition that the \emph{maximum} of time-$0$ portfolio value needs to be zero ($\max\limits_{\omega \in \Omega}\{V_{0}(x_{0},\Delta)\left(\omega\right)\}=0$). It is obvious that in the case of complete information (i.e, $H=0$), this definition boils down to the classical definition of arbitrage opportunity. \par
 We need to show that there is no arbitrage in our discrete time model with delayed information. \cite{kabanov2006dalang} proves that in a general discrete time model with restricted information, there does not exist classical arbitrage, if and only if there exists a probability measure $\widetilde{\mathbb{P}}$ equivalent to $\mathbb{P}$ such that the optional projection under $\widetilde{\mathbb{P}}$ of the discounted stock price on the delayed filtration, is a $\widetilde{\mathbb{P}}$-martingale The setup of our model is a bit different than that in  \cite{kabanov2006dalang}, given that our first order to buy/ sell the risky asset is executed at time $H$, rather than at time $0$ (i.e. $\Delta_{k}=0$, $k=0,\dots,H-1$). This makes the initial portfolio value ($V_{0}(x_{0},\Delta)$) a random variable, rather than always a constant. Theorem \ref{NACondition} shows that still in our model, there does not exists arbitrage, in the sense of Definition \ref{arbitrage}.  
 \begin{theorem}
 \label{NACondition}
 There does not exists any arbitrage opportunity in our discrete time model, in the domain of $(\mathbb{R},\mathcal{A}_{\mathfrak{G}})$ strategies.
 \end{theorem}
 \begin{proof}
 \label{arbitrageproof}
 According to Definition \ref{arbitrage}, absence of arbitrage means that for any strategy $(x_0,\Delta) \in (\mathbb{R},\mathcal{A}_{\mathfrak{G}})$ such that $\max\limits_{\omega \in \Omega}\{V_{0}(x_{0},\Delta)\left(\omega\right)\} = 0$, the condition $\mathbb{P}(V_{N}(x_{0},\Delta)\geq 0) = 1$ implies that $\mathbb{P}(V_{N}(x_{0},\Delta)= 0) = 1$.\par
 In the domain of $(\mathbb{R},\mathcal{A}_{\mathfrak{G}})$, according to (\ref{initialportfolio}), the condition $\max\limits_{\omega \in \Omega}\{V_{0}(x_{0},\Delta)\left(\omega\right)\} = 0$ is equivalent to
 \begin{eqnarray}
 \max\limits_{\omega \in \Omega}\{V_{H}(x_{0},\Delta)\left(\omega\right)\} &=& 0,  \nonumber
 \end{eqnarray}
which means that in all $(N-H)$-period binomial models starting from time $H$, the initial values for the $(x_0,\Delta)$ strategy are non-positive.\par 
 
 If we consider all these $(N-H)$-period binomial models individually, they lie in the general discrete time model framework in \cite{kabanov2006dalang}. Therefore, in each of these models, even if we consider the initial values of the strategy to be zero, the condition $\mathbb{P}(V_{N}(x_{0},\Delta)\geq 0) = 1$ implies $\mathbb{P}(V_{N}(x_{0},\Delta)= 0) = 1$, given that we show that there exists a probability measure $\widetilde{\mathbb{P}} \sim \mathbb{P}$ such that the $\widetilde{\mathbb{P}}$-optional projection of the discounted stock price on the delayed filtration, is a $\widetilde{\mathbb{P}}$-martingale,  that is 
 \begin{eqnarray}
 \label{riskmeasure}
 \mathbb{E}^{\widetilde{\mathbb{P}}}\left(\widetilde{S}_{k+1} \lvert \mathcal{G}_{k} \right)=\mathbb{E}^{\widetilde{\mathbb{P}}}\left(\widetilde{S}_{k} \lvert \mathcal{G}_{k} \right), \quad k=H,\dots,N-1. 
 \end{eqnarray}  
 Define the probability measure $\widetilde{\mathbb{P}}$ such that the coordinate maps $Z_{k}, k = 1, \ldots , N$ are still independent Bernoulli random variables, but with parameters 
 \begin{eqnarray}
 \mathbb{\widetilde{P}}(Z_{k}=1)=\frac{u e^{r}-d}{u-d}=1-\mathbb{\widetilde{P}}(Z_{k}=0), \quad k \in \{1,\dots,N\}, \nonumber  
 \end{eqnarray}
 which are the risk-neutral probabilities in the usual binomial model without any delay.\par
 Given that the discounted stock price ($\widetilde{S}_{k}$) is $(\mathcal{F}_{k})$-martingale under $\widetilde{\mathbb{P}}$, it follows that condition (\ref{riskmeasure}) holds, which shows that there is no arbitrage opportunity from time $H$ to $N$. Consequently, given (\ref{initialportfolio}), we conclude that there is no arbitrage in the model in the domain of $(\mathbb{R},\mathcal{A}_{\mathfrak{G}})$ strategies.
 \end{proof}
 \begin{remark}
 The domain of $(\mathbb{R},\mathcal{A}_{\mathfrak{G}})$ strategies in Theorem (\ref{NACondition}) does not include all $\mathfrak{F}$-adapted strategies, but only those which are $\mathfrak{G}$-adapted. In other words, we are excluding the case that an agent with full information comes and exploits the advantage over the investors with delayed information in the market. If we include all $\mathfrak{F}$-adapted strategies, it is likely to have arbitrage opportunities.
 \end{remark}
 \subsection{Super-Replication Price}
 Given that there is no arbitrage in the market, it now makes sense to discuss about pricing.  
 \begin{definition}[Super-replication price and the value process of super-replicating portfolio] 
 \label{superreplication}
 For any contingent claim with payoff function $\, \varphi : \Omega \to \mathbb R \, $ and expiration time $\,N \,$,
 its super-replication price $\, \bar{\pi}(\varphi)\, $ is defined as the minimal initial value of portfolio which exceeds the value $\,\varphi \,$ at time $\,N\,$, i.e., 
 \begin{eqnarray} \label{superreplicationequ}
 \bar{\pi}(\varphi) \, :=\,  \inf_{(x_{0},\Delta) \in \Gamma} \max_{\omega \in \Omega}{\left\{V_{0}(x_{0},\Delta)(\omega) \, =\, x_{0}+\Delta_{H}\widetilde{S}_{H}(\omega) \right\}} \, , 
  \end{eqnarray}
 where 
 \begin{eqnarray}
 \label{barpi}
 \, \Gamma \, :=\,  \left\{(x_{0},\Delta) \in \mathbb{R} \times \mathcal{A}_{\mathfrak{G}} : \, V_{N}(x_{0},\Delta) \geq \varphi \hspace{2.5 mm}  \mathbb{P}-a.s. \right\}\,.
 \end{eqnarray}
 If there exists a pair $\, (x_{0}^{\ast},\Delta^{\ast})\, $ that attains the infimum in (\ref{superreplicationequ}), i.e.,  $\bar{\pi}(\varphi) \, =\, \max_{\omega \in \Omega} V_{0}(x_{0}^{\ast},\Delta^{\ast})(\omega)\, $, then the time-$\, k\, $ super replicating portfolio value $\, \mathcal{V}_{k} (\omega)\, $ 
 is defined as  
\begin{eqnarray} \label{superrepvalueprocess}
 \mathcal{V}_{k} (\omega) \, :=\,  V_{k}(x_{0}^{\ast},\Delta^{\ast}) (\omega) \,, \quad \, k \, =\,  0, \ldots , N\,,  
\end{eqnarray}
 and consequently, $\, 
 \bar{\pi}(\varphi) \, =\, \max_{\omega \in \Omega} \mathcal{V}_{0}(\omega)\, $.  
 \end{definition}
 
 \begin{remark}
  The super-replication price is the most conservative pricing approach for the seller of the option, considering the worst-case scenario. In other words, it is straightforward to show that any price greater than the super-replication price causes arbitrage in the market.
  \end{remark}
  
  \begin{remark}
   \label{callputparity}
    It is remarkable to note that \textit{call-put parity} does not hold anymore. The reason is that the super-replication price $\, \bar{\pi} \, $ is a coherent risk measure on the space $\, \mathbb{L}^{\infty}(\Omega,\mathfrak{F},\mathbb{P})\, $ of payoff functions, and therefore it is subadditive.
    \end{remark}
 
%
 
  All of the results in this paper are for European-style contingent claims with convex payoff functions. In section \ref{sec1}, we consider first the case $\, H \, =\, N-1\, $ and determine the super-replication price and the corresponding strategy. This would make the building block for the general case discussed in section \ref{sec2}. The case for non-convex payoff functions is computationally more demanding as we do not have access to all the machinery developed for convex functions.

 \subsection{An $N$-period binomial model with $H=N-1$ periods of delay}
 \label{sec1}
   We determine the super-replication price and the corresponding strategy for the European contingent claims when $\, H \, =\, N-1\, $. Having $\, H \, =\, N-1\, $ periods of delayed information means that at time $\, 0\, $  the risky asset price $\, S_0\, $ is observed, but the order $\, \Delta_{H}\, $, sent by the investor at time $\,0\,$, would be executed at time $\, H\, $.  For example, when $\, N \, =\, 2\, $ and $\,H \, =\, 1\,$, the order $\,\Delta_{1}\,$ sent at time $\,0\,$ is executed at time $\,1\,$ with two possible prices $\, S_1 \, =\, S_0 d\, $ or $\, S_{1} \, =\, S_0 u\, $ (see Figure \ref{fig1}). 
 
 Let us observe that in the case of $\,H \, =\, N-1\,$, the terminal value $\,V_{N}(x_{0}, \Delta)\,$  in (\ref{portfolio}) is simplified to 
  \begin{equation} \label{VN}
  V_{N}(x_{0}, \Delta)(\omega) \, =\, e^{rN}x_{0}+ S_{N} (\omega) \cdot \Delta_{N-1}  \,. 
  \end{equation}
 There are $\, (N+1)\,$ possible values of $\, S_{N}(\omega)\,$, $\, \omega \in \Omega\,$ in (\ref{priceprocess}) and there are only two controls $\,(x_{0}, \Delta_{N-1}) \,$ in the terminal value.  
 Since there are $\, (N+1)\, $ constraints and only two controls, the minimization problem in (\ref{superreplicationequ}) has possibly infinitely many solutions. In other words, in an economic sense, the market is not complete. To learn more about pricing in incomplete markets, we refer to \cite{staum2007incomplete}.

    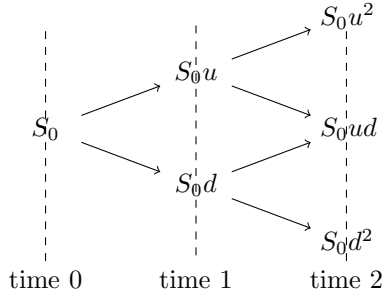
\begin{figure}
 \centering
 \tikzstyle{bag} = [text width=2em, text centered]
 \tikzstyle{end} = []
 \begin{tikzpicture}[sloped]
   \node (a) at ( 0,0) [bag] {$S_0$};
   \node (b) at ( 2,-0.75) [bag] {$S_0 d$};
   \node (c) at ( 2,0.75) [bag] {$S_0 u$};
   \node (d) at ( 4,-1.5) [bag] {$S_0 d^2$};
   \node (e) at ( 4,0) [bag] {$S_0 u d$};
   \node (f) at ( 4,1.5) [bag] {$S_0 u^2$};
   \node (aa1) at (0,-2){time 0};
   \node (aa2) at (0,1.5){};
   \node (bb) at (2,-2.0){time 1};
   \node (cc) at (2,1.5){};
   \node (dd) at (4,-2){time 2};
   \node (ff) at (4,1.5){};
   \draw [->] (a) to node [below] {} (b);
   \draw [->] (a) to node [above] {} (c);
   \draw [->] (c) to node [below] {} (f);
   \draw [->] (c) to node [above] {} (e);
   \draw [->] (b) to node [below] {} (e);
   \draw [->] (b) to node [above] {} (d);
   \draw [dashed] (aa1) to node [above]{} (aa2);
   \draw [dashed] (cc) to node [above]{} (bb);
   \draw [dashed] (dd) to node [above]{} (ff);
 \end{tikzpicture}
 \caption{Asset price process $S_{k}$ in a $2$-period binomial model}
 \label{fig1}
 \end{figure}

 \begin{theorem}
 \label{thm1}
 For a European-style contingent claim with payoff $\, \varphi \, :=\, \Phi(S_{N}) \, $ for some convex function $\, \Phi(\cdot)\,$ in the $\, N$-period binomial model with $\, H \, =\, N-1\, $ periods of delay, the super-replication price is
 \begin{eqnarray}
 \bar{\pi}(\varphi) \, =\, \max{\left( x_{0}^{\ast}+e^{-rH}\,  \Delta_{H}^{\ast} \cdot \, S_{0}\, u^{H}\, , \, x_{0}^{\ast}+e^{-rH} \, \Delta_{H}^{\ast} \cdot \, S_{0}\, d^{H} \right)}\, , 
 \end{eqnarray}
 where the corresponding strategy $\, (x_{0}^{\ast}, \Delta^{\ast}) \,$ is given by $\, \Delta_{j}^{\ast} \equiv 0\,$, $\,j \, =\, 0, 1, \ldots , H-1\,$, 
 \begin{equation} \label{SH}
 \Delta_H ^{\ast} \, =\, \Delta_{N-1}^{\ast} \, =\, \frac{\, \Phi(S_0 u^{N}) - \Phi (S_0 d^{N})\, }{\, S_0  \cdot (u^{N}- d^{N}) \, }
 \quad
 \text{ and } \quad
 x_{0}^{\ast}=e^{-rN}\cdot \frac{\, u^{N} \Phi (S_0 d^{N})-  d^{N} \Phi (S_0 u^{N})\, }{ u^{N} - d^{N} }\, .
 \end{equation}
 \end{theorem}
 \begin{proof}
 First, we shall prove that for any $\, \omega \in \Omega\, $, $\, (x_{0}^{\ast}, \Delta^{\ast})\, $ in (\ref{SH}) satisfies 
 \begin{eqnarray}
 \label{claim}
 \inf_{(x_{0},\Delta) \in \Gamma} \left\{V_{0}(x_{0},\Delta)(\omega) \, =\, x_{0}+\Delta_{H}\widetilde{S}_{H}(\omega) \right\} \, =\,  V_{0}(x_{0}^{\ast},\Delta^{\ast})(\omega) \, =\,  x_{0}^{\ast}+\Delta_{H}^{\ast} \cdot \widetilde{S}_{H}(\omega).
 \end{eqnarray}
 Here the infimum is taken over the set $\, \Gamma\, $ in (\ref{barpi}), that is, $\, x_{0} \in \mathbb{R}\, $ and $\, \Delta \in \mathcal{A}_{\mathfrak{G}}\, $ must satisfy $\, V_{N}(x_{0},\Delta) \geq \varphi(S_{N}) \, $ almost surely. Note that $\, V_{N}(x_{0},\Delta) \, =\,  (e^{rN}x_{0}+ x \cdot \Delta_{N-1})\vert_{x = S_{N}}\, $ in (\ref{VN}) is realized as the value at $\,x = S_{N}\,$ of linear function $ \, y \, =\,  e^{rN} \, x_{0}+ x \cdot \Delta_{N-1}\, $ with the slope $\, \Delta_{H}\,$ and the $\,y\,$-intercept $\, e^{rN} x_{0} \,$ in the $\, (x, y)\, $ coordinates. Moreover, since the payoff function $\,  \Phi(\cdot)\, $ is convex, by Jensen's inequality, one can verify 
  \begin{eqnarray}
 \label{endpoint}
 \Gamma \, =\, \big\{ (x_{0}, \Delta) \in \mathbb R \times \mathcal A_{\mathfrak G} \, : \, e^{rN}x_{0}+S_{0}u^{N} \cdot \Delta_{N-1} \, \geq\,  \Phi(S_{0}u^{N}) ,  \quad 
 e^{rN}x_{0}+S_{0}d^{N} \cdot \Delta_{N-1} \, \geq \, \Phi(S_{0}d^{N}) \} \, . \nonumber \\
 \end{eqnarray}
 That is, in order to check whether the inequality $\, V_{N}(x_{0}, \Delta) \, \ge \, \Phi (S_{N}) \, $ holds with probability one, it suffices to check it just at the extreme cases, in which the asset price $S_{N}$ at time $N$ is the minimum $ S_{0} d^{N}$ or the maximum $\, S_{0} u^{N}\,$ in the binomial tree model. Then it is easy to check that the choice $ (x_{0}^{\ast}, \Delta_{H}^{\ast})$ in (\ref{SH}) belongs to the set $\, \Gamma\,$ as we have $\,  e^{rN} x_{0}^{\ast}+  \Delta_{H}^{\ast} {S}_{0}  u^{N} \, =\, \varphi(S_{0}u^{N}) \, $, $\, e^{rN} x_{0}^{\ast}+ \Delta_{H}^{\ast}  S_{0} d^{N} \, =\, \varphi(S_{0}d^{N}) \, $. 
 In other words, the minimization problem is reduced to a linear programming problem
 \begin{eqnarray}
 && \underset{(x_{0},\Delta_{H}) \in \mathbb R^{2}}{\text{minimize}}
 \,\, x_{0}+\Delta_{H} \cdot \widetilde{S}_{H}(\omega) \nonumber \\
 && \text{subject to}
 \,\,  \quad e^{rN}x_{0}+S_{0}u^{N} \cdot \Delta_{H} \, \geq \, \Phi(S_{0}u^{N})\, ,  \quad \text{ and }  \quad   e^{rN}x_{0}+S_{0}d^{N} \cdot \Delta_{H} \, \geq \, \Phi(S_{0}d^{N})\, .  \nonumber
 \end{eqnarray}
 Define the Lagrangian as
 \begin{eqnarray*}
 \mathcal{L} 
  \, :=\,  x_{0}+\Delta_{H}\widetilde{S}_{H}(\omega) + \lambda_{1} [\Phi(S_{0}u^{N})-\left(e^{rN}x_{0}+S_{0}u^{N}\Delta_{H}\right) ] 
 + \lambda_{2} [\Phi(S_{0}d^{N})-\left(e^{rN}x_{0}+S_{0}d^{N}\Delta_{H} \right) ]\, , 
 \end{eqnarray*}
 where $\, \lambda_{1}\, $ and $\, \lambda_{2}\, $ are the Lagrangian multipliers. Then, it is easy to check that the quantities 
 \begin{eqnarray*}
 x_{0}^{\ast}\, =\, e^{-rN}\cdot \frac{\, u^{N} \Phi(S_0 d^{N})-  d^{N} \Phi(S_0 u^{N})\, }{u^{N} - d^{N}}\, ,  \quad 
 \Delta_H ^{\ast} \, =\,  \frac{\Phi(S_0 u^{N})-\Phi(S_0 d^{N})}{S_0 u^{N}-S_0 d^{N}}\, , \\
 \lambda_{1}^{\ast} \, =\,  \frac{\, \widetilde{S}_{H}(\omega)-e^{-rN}S_{0}d^{N}\, }{S_{0}\cdot (u^{N}-d^{N})}\, , \quad  
 \lambda_{2}^{\ast} \, =\,  \frac{e^{-rN}S_{0}u^{N}-\widetilde{S}_{H}(\omega)}{S_{0}u^{N}-S_{0}d^{N}}\,   \nonumber
 \end{eqnarray*}
 satisfy the Karush-Kuhn-Tucker conditions for the minimization. Hence, (\ref{claim}) follows, and it is the key to prove that  
 \begin{equation}
  \inf_{(x_{o},\Delta) \in \Gamma} \max_{\omega \in \Omega}{ V_{0}(x_{0},\Delta)(\omega) } \, =\,  \max_{\omega \in \Omega}\inf_{(x_{o},\Delta) \in \Gamma} V_{0}(x_{0},\Delta)(\omega)  \,. 
 \end{equation}
 Thus, we get
 \begin{equation}
 \label{www}
  \bar{\pi}(\varphi)\, =\,\max_{\omega \in \Omega}\inf_{(x_{o},\Delta) \in \Gamma} V_{0}(x_{0},\Delta)(\omega)  \, =\, \max\limits_{\omega \in \Omega} V_{0}(x_{0}^{\ast},\Delta^{\ast})(\omega)\, .
 \end{equation}
 Then, the proof is completed by the following observation 
  \begin{eqnarray}
 \label{maximum}
 \max_{\omega \in \Omega} {\big \{V_{0}(x_{0}^{\ast},\Delta^{\ast})(\omega)=x_{0}^{\ast}+\Delta_{H}^{\ast}\widetilde{S}_{H}(\omega)\big \}=\max{\left( x_{0}^{\ast}+e^{-rH} \Delta_{H}^{\ast} \cdot S_{0}u^{H},x_{0}^{\ast}+e^{-rH} \Delta_{H}^{\ast} \cdot S_{0}d^{H} \right)}.} \nonumber \\
 \end{eqnarray}
 \end{proof}
   By using Theorem \ref{thm1}, the portfolio value $\mathcal{V}_{H} \in \mathfrak{F}_{H}$ in (\ref{superrepvalueprocess}) at time $H$ of the super-replicating strategy can be calculated as
 \begin{eqnarray}
 \label{rep}
 \mathcal{V}_{H} \, \, =\, 	\, e^{rH}x_{0}^{\ast}+\Delta_{H}^{\ast} \cdot S_{H} \nonumber 
                             \, &=&\,  \sum_{j=0}^{H}e^{-r (N-H)} \, \mathbb{E}^{\mathbb{Q}_{j}} \big[ \Phi(S_{N}) \big]\cdot \mathbbm{1}_{\{ S_{H} \, =\,  S_{0}\, u^{j}\, d^{H-j}\}}  \nonumber \\
                             \, \, &=&\, \, \sum_{j=0}^{H}e^{-r(N-H)} \left [ \, \mathfrak{p}_{j}\,  \Phi(S_{0}u^{N}) + \mathfrak{q}_{j} \, \Phi(S_{0}d^{N})   \, \right]  \cdot \mathbbm{1}_{\{S_{H} \, =\,  S_{0}\, u^{j}\, d^{H-j}\}}.
 \end{eqnarray}
 Here $\{\mathbb{Q}_{j}\}_{j=0}^{H}$ are probability measures on $(\Omega,\mathfrak{F})$ defined by 
 \begin{eqnarray}
 \label{measure}
 \mathbb{Q}_{j}(I_{N}\, =\, N)\, :=\, \mathfrak{p}_{j} \, =\, 1-\mathbb{Q}_{j}(I_{N}\, =\, 0) \, =\, 1-\mathfrak{q}_{j} \, , 
 \quad \mathfrak{p}_{j} \, :=\, \frac{\, u^{j}d^{H-j}e^{r}-d^{H+1}\, }{\, u^{H+1}-d^{H+1}\, }\, , \hspace{2 mm} j \, =\, 0,\ldots,H\, . \nonumber \\
 \end{eqnarray}
 \begin{remark} \label{pathdependent}
 We can conclude from the form in (\ref{rep}) that $\,\mathcal{V}_{H} \,$, the value of the super-replicating portfolio at time $H$, is a function of $\, S_0 \,$ and $\, S_H \,$. In other words
 \begin{equation}
 \mathcal{V}_{H} \equiv \mathcal{V}_{H}(S_0,S_H). \nonumber
 \end{equation}  
 Therefore, the value process for the super-replicating portfolio is \emph{path dependent}, due to the existence of $H$ periods of lag between the times of order submission and execution.
 \end{remark}
 Thus, the super-replication price $\, \bar{\pi}(\varphi)\, $ can be calculated as
 \begin{eqnarray} \label{superrepeasy}
 \bar{\pi}(\varphi)\,=\,\max\limits_{\omega \in \Omega} \mathcal{V}_{0} \left(\omega\right)  
                   \, &=&\, \max_{j \in \{0,\dots,H\}}e^{-rN} \, \mathbb{E}^{\mathbb{Q}_{j}}\big [\Phi(S_{N})\big]    \nonumber
                   \, =\, \max_{j \in \{0,H\}} e^{-rN}\, \mathbb{E}^{\mathbb{Q}_{j}}\left [\Phi(S_{N}) \right]  \\
                   \,&=&\,\max_{j \in \{0,H\}}e^{-rN} \left[ \mathfrak{p}_{j} \Phi(S_{0}u^{N}) + \mathfrak{q}_{j} \Phi(S_{0}d^{N})  \right]  \nonumber \\
                   \,&=&\, e^{-rN} \max \left( \mathfrak{p}_{u} \Phi(S_{0}u^{N}) + \mathfrak{q}_{u} \Phi(S_{0}d^{N}) \, , \, \mathfrak{p}_{d} \Phi(S_{0}u^{N}) + \mathfrak{q}_{d} \Phi(S_{0}d^{N})  \right) \,,
 \end{eqnarray}   
  where the third equality follows similarly as in (\ref{maximum}). \par
  
  \medskip 
  
 \noindent \textbf{Notation:} From now on, we use $\, (\mathfrak{p}_{u},\mathfrak{q}_{u})\, $ as $\, ( \mathfrak{p}_{H}, \mathfrak{q}_{H})\, $ and $\, (\mathfrak{p}_{d},\mathfrak{q}_{d})\, $ as $\, (\mathfrak{p}_{0},\mathfrak{q}_{0})\, $, since $\, (\mathfrak{p}_{H},\mathfrak{q}_{H})\, $ and $\, (\mathfrak{p}_{0},\mathfrak{q}_{0})\, $ correspond to the measures at the extreme points $\, S_{H} \, =\, S_{0}u^{H}\, $ and $\, S_{H} \, =\, S_{0}d^{H}\, $  respectively.

 \subsection{An $\, N$-Period binomial Model with $\, H\, $ Periods of Delay}
 \label{sec2}
 We extend our considerations from section \ref{sec1} and generalize the model to the $\, N$-period binomial model with $\, H (\le N-1)\, $ periods of delay.  We determine the super-replication price and the corresponding strategy for European style contingent claims with convex payoff functions. Here we shall solve the problem from both a dynamic programming (or backward induction) approach and a direct approach.
  
 \subsubsection{Dynamic Programming Approach}
 \label{dynamic} 
 First, let us define the tree $\, {\mathcal T}_{N}(0, 0) \,$ of length $\, N\,$ as the set of nodes $\, (i,j) \,$, such that there are $\,i \,$ ups and $\, j\,$ downs from the node $\, (0, 0) \,$ with $\, 0 \le i+j \le N \,$, i.e., 
  \[
  {\mathcal T}_{N} (0, 0) \, :=\, \{ (i,j) \in \mathbb N_{0}^{2}\, : \, 0 \le i + j \le N \} .
  \]
  Then define its $(H+1)$-period subtree $\, \mathcal T_{H+1} (a, b) \,$ starting from the node $\, (a, b) \,$ at time $\, a+b\,$ by 
  \[
  \mathcal T_{H+1} (a,b) \, :=\, \{ (i,j) \in \mathbb N_{0}^{2} \,: \, a+b\le i+j \le a+b+H+1\, ,\,  i \ge a \, , \, j \ge b\}  ,
  \]
  for every $(a,b) \in {\mathcal T}_{N} (0, 0)$ such that $a+b \leq N-(H+1)$. \par 
  
 We shall identify all $\, N-H\, $ subtrees $\mathcal T_{H+1} (a,b)$ starting from the  nodes $\, (a, b) \,$ at time $\, N-(H+1) \,$ (i.e, $\, a+b=N-(H+1)\,$). We use the results in section \ref{sec1} and consider the value process of the super-replicating portfolio at time $N-1$ as the new payoffs for the next round of ($H+1$)-period subtrees starting from the nodes at time $\, N-(H+2) \,$. Then, we keep super-replicating backwards in the same manner. \par 
 
 \begin{remark}
 Given that in the dynamic programming approach, we are using the results in section (\ref{sec1}) in each step, and Remark (\ref{pathdependent}), we can conclude that the value process at level $k \in \{H,\dots, N\}$ for the super-replicating strategy in the general model is also path dependent, that is
 \begin{equation} \label{supereplvalueprocess22}
 \mathcal{V}_{k}\,\equiv\,\mathcal{V}_{k}(S_{k-H},S_k), \, \quad  k \, =\,  H,\dots,N\,. 
 \end{equation}
 \end{remark}
 
 Therefore, let us define the payoff for the subtree $\, \mathcal T_{H+1} (a, b) \,$ starting from the node $\, (a, b) \,$ at time $\, a+b\,$ at its leaf node $\,(p,q) \,$ (i.e, $\, p+q = a+b+H+1 \,$)  by 
   \begin{equation}
   \label{payoff}
   \\ \numberthis
    \Phi_{\mathcal T_{H+1} \left(a, b\right)}\left(p,q\right) \, :=\,  
   \left \{
   \begin{array}{ll}
   \mathcal{V}_{p+q}\left(\, S_{a+b} d\, ,\, S_{a+b} d^{H+1}\,\right) \, \quad \text{ if } p\, =\,  a\,;\\
   \\
   \max\bigg\{\mathcal{V}_{p+q}\left(\, S_{a+b} d\, ,\, S_{a+b} u^{i}d^{H+1-i} \,\right),\mathcal{V}_{p+q}\left(\, S_{a+b} u\,,\, S_{a+b} u^{i}d^{H+1-i} \, \right)\bigg\}  \, \quad  \text{ if }  p\,=\, a+i\,, \\
   \quad        \, \hspace{10.7 cm}  i \,=\,1,\ldots,H\, ;  \nonumber \\ 
   \mathcal{V}_{p+q}\left(\, S_{a+b} u\, ,\, S_{a+b} u^{H+1} \, \right) \, \quad \text{ if } p\, =\,  a+H+1\,;  \\
   \end{array}
   \right . 
   \end{equation}
 for $\,p+q \, \leq\, N-1$, and 
  $\,  \Phi_{\mathcal T_{H+1} \left(a, b\right)}\left(p,q\right):= \Phi(S_{N})\, $, 
  $\, p+q \, =\, N\, $ where $S_N\,=\,S_0 u^p d^q$. \par 
  Intuitively, for the subtree $\mathcal T_{H+1} \left(a, b\right)$ starting at time $\, a+b\,$, there are only two $\, (H+1)\, $-period subtrees, $\mathcal T_{H+1} \left(a+1, b\right)$ and $\mathcal T_{H+1} \left(a, b+1\right)$, starting at time $\, a+b+1\,$ that can induce payoff at time $\, p+q\,$. So, we need to take the maximum of the two possible value process as the new payoff because we always consider worst case scenario in super replication. Note that at the edge points, there exists only one value process. \par 
 
 \begin{example} In the $\, 4\, $-period binomial tree model (as in Figure \ref{fign4}) with $\, H=1\, $, what new payoff we need to consider on the node $\, S_{3}\, =\,  S_{0} u^{2}d\, $ depends on whether we are considering this node as part of the subtree $\mathcal T_{2} \left(1, 0\right)$ or $\mathcal T_{2} \left(0, 1\right)$. As part of the subtree $\mathcal T_{2} \left(1, 0\right)$, the payoff $\left(\Phi_{\mathcal T_{2} \left(1,0\right)}\left(2,1\right)\right) \,$ would be the maximum of the corresponding value processes of the subtrees $\mathcal T_{2} \left(1, 1\right)$ and $\mathcal T_{2} \left(2, 0\right)$, while as part of the subtree $\mathcal T_{2} \left(0, 1\right)$, the payoff $\left(\Phi_{\mathcal T_{2} \left(0,1\right)}\left(2,1\right)\right) \,$ would be the corresponding value processes of the subtrees $\mathcal T_{2} \left(1, 1\right)$.
  \end{example}

 One important ingredient in the dynamic programming approach is that when we start from a convex payoff function, the payoff in (\ref{payoff}) for all the intermediary $\, (H+1)\,$-period subtrees needs to be convex with respect to the corresponding risky asset prices, in order to be able to use Theorem (\ref{thm1}) in each step and keep super-replicating backwards. Theorem (\ref{thm2}) formalizes this relation.\par

 \begin{theorem}
 \label{thm2}
 For a European-style contingent claim with payoff $\, \varphi \, :=\, \Phi(S_{N}) \, $ for some convex function $\, \Phi(\cdot)\,$ in the $\, N$-period binomial model with $\, H \, \leq\, N-1\, $ periods of delay, the payoff function $\Phi_{\mathcal T_{H+1} \left(a, b\right)}\left(.,.\right)$, $a+b\, =\, 0,\dots,N-(H+1)$ in (\ref{payoff}) for all the intermediary $(H+1)$-period subtrees are convex with respect to the corresponding risky asset prices.
 \end{theorem}
 \begin{proof}
 Note that for $a+b\,=\, N-(H+1)$, the payoff functions $\,  \Phi_{\mathcal T_{H+1} \left(a, b\right)}\left(.,.\right)$ for all $(N-H)$ intermediary $(H+1)$-period subtrees are convex, since the final payoff function $\Phi(S_{N})$ is convex.\par
 
 Now we show that all the payoff functions $\,  \Phi_{\mathcal T_{H+1} \left(a', b'\right)}\left(.,.\right)$, $a'+b'=a+b-1$ will be convex, if all the payoff functions $\,  \Phi_{\mathcal T_{H+1} \left(a, b\right)}\left(.,.\right)$, $a+b \in \{0,\dots,N-(H+1)\}$ are convex. By induction this completes the proof.\par
 
 Given that the payoff function $\,  \Phi_{\mathcal T_{H+1} \left(a, b\right)}\left(.,.\right)$ is convex, by Theorem (\ref{thm1}), there exists $x^{\ast}_{1}$ and $\Delta^{\ast}_{1}$ such that we define
  \begin{eqnarray}
  \left .
  \begin{array}{l}
  h_1(t):=\\
     \end{array}
     \right . 
     \left \{
     \begin{array}{ll}
     \mathcal{V}_{a+b+H}(\,S_{a'+b'}u\,, \,S_{a+b+H}\,)=e^{rH}x^{\ast}_{1}+\Delta^{\ast}_{1}t, \quad t \in \{S_{a'+b'}ud^{H},\dots,S_{a'+b'}u^{H+1}\}; \\
     \\
     e^{rH}x^{\ast}_{1}+\Delta^{\ast}_{1}t, \quad t=S_{a'+b'}d^{H+1}. \nonumber
     \end{array}
        \right .
  \end{eqnarray}
 Similarly, there exists $x^{\ast}_{2}$ and $\Delta^{\ast}_{2}$ such that we define
 \begin{eqnarray}
   \left .
   \begin{array}{l}
   h_2(t):=\\
      \end{array}
      \right . 
      \left \{
      \begin{array}{ll}
      \mathcal{V}_{a+b+H}(\,S_{a'+b'}d\,,\,S_{a+b+H}\,)=e^{rH}x^{\ast}_{2}+\Delta^{\ast}_{2}t, \quad t \in \{S_{a'+b'}d^{H+1},\dots,S_{a'+b'}u^{H}d\}; \\
      \\
      e^{rH}x^{\ast}_{2}+\Delta^{\ast}_{2}t, \quad t=S_{a'+b'}u^{H+1}, \nonumber
      \end{array}
         \right .
   \end{eqnarray}
	we can define
   \begin{eqnarray} \label{hdef}
   h(t)\,:=\,  \max\left(h_{1}(t),h_{2}(t)\right), \quad   \,  t \in \{S_{a'+b'}d^{H+1},\dots, S_{a'+b'}u^{H+1}\};  
   \end{eqnarray}
   Note that $h(t)\,=\,\Phi_{\mathcal T_{H+1} \left(a', b'\right)}\left(p,q\right)$ where $t\,:= \, S_0 u^p d^q$, given that $\,  \Phi_{\mathcal T_{H+1} \left(a, b\right)}\left(.,.\right)$ is convex, and (\ref{payoff}).\par 
 The discrete function $h(.)$ is convex if for any $v$ and $w$ such that $S_0 u^v d^w \in \{S_{a'+b'}ud^{H},S_{a'+b'}u^{H}d\}$, we have 
  \begin{eqnarray}
  \label{convex}
  h\left(t_p\right) + h\left(t_n\right) \geq 2 h\left(t_m\right), 
  \end{eqnarray}
  where $t_p:\,=\, S_0 u^{v-1} d^{w+1}$, $t_n:\,=\, S_0 u^{v+1} d^{w-1}$ and $t_m:\,=\, S_0 u^{v} d^w$.\par 
  Depending on the choice of $v$ and $w$, there are 4 cases:\\
   Case 1: $h(t_{p})=h_{1}(t_{p})$ and $h(t_{n})=h_{1}(t_{n})$. Then, given the form in (\ref{hdef}), we have $h(t_m)=h_1(t_m)$. Then, it is straightforward to show that (\ref{convex}) follows by linearity of the function $h_1(t_m)$.\\
   Case 2: $h(t_{p})=h_{2}(t_{p})$ and $h(t_{n})=h_{2}(t_{n})$. This case follows similar to that of case 1.\\
   Case 3: $h(t_{p})=h_{1}(t_{p})$ and $h(t_{n})=h_{2}(t_{n})$. Then, $h(t_m)$ would equal to either $h_{1}(t_m)$ or $h_{2}(t_m)$. Without loss of generality assume that $h(t_m)\, =\, h_{1}(t_m)$. Then given that $h(t_{n})=h_{2}(t_{n})$, we conclude by the form in (\ref{hdef}) that $h_{2}(t_{n}) \geq h_{1}(t_{n})$. So, we derive 
   \begin{eqnarray}
 	h\left(t_p\right) + h\left(t_n\right)\,=\, h_1\left(t_p\right) + h_2\left(t_n\right) 
 										\, \geq \, h_1\left(t_p\right) + h_1\left(t_n\right)
 										\, \geq \, 2 h_1\left(t_m\right) 
 										\, =\, 2 h\left(t_m\right) ,  \nonumber
   \end{eqnarray}
   where the last inequality follows by the linearity of the $h_{1}(.)$ function.\\
   Case 4: $h(t_{p})=h_{2}(t_{p})$ and $h(t_{n})=h_{1}(t_{n})$. This case follows similar to that of case 3.
 \end{proof}

 Therefore, given Theorem (\ref{thm2}), we can apply the dynamic programming approach, and derive the portfolio value $\mathcal{V}_{k}(S_{k-H},S_{k})$ in (\ref{supereplvalueprocess22}) at level $k\, = \, a+b +H$, $k \in \{H,\dots,N-1\}$ of the super-replicating strategy, using representation (\ref{rep}), as
 \begin{eqnarray}
  \label{dp}
  \mathcal{V}_{k}(S_0 u^a d^b,S_{k}) \,=\, \sum\limits_{j=0}^{H}e^{-r}\bigg[ \mathfrak{p}_{j}\Phi_{\mathcal T_{H+1} \left(a, b\right)}\left(a+H+1,b\right)+\mathfrak{q}_{j}\Phi_{\mathcal T_{H+1} \left(a, b\right)}\left(a,b+H+1\right)\bigg] \mathbbm{1}_{\{S_{k} = S_{k-H}u^{j}d^{H-j}\}}, \nonumber \\
  \\
  \,\,k=H,\dots,N-1, \nonumber
  \end{eqnarray}
  where $\mathfrak{p}_j$ and $\mathfrak{q}_j$, $j=0,\dots,H$ are defined as in (\ref{measure}).\par 
 Plugging in (\ref{payoff}) for $k=H,\dots,N-2$, we obtain the key recursive formula
 \begin{eqnarray}
 \label{recur}
 \mathcal{V}_{k}(S_{k-H},S_{k}) \,=\,\sum\limits_{j=0}^{H}e^{-r}\bigg[\mathfrak{p}_{j}\mathcal{V}_{k+1}(\,S_{k-H}u\,,\,S_{k-H}u^{H+1})+ \mathfrak{q}_{j}\mathcal{V}_{k+1}(\,S_{k-H}d\,,\,S_{k-H}d^{H+1}\,) \bigg]\mathbbm{1}_{\{S_{k} = S_{k-H}u^{j}d^{H-j}\}},\nonumber \\
 \\
 \,\,k=H,\dots,N-2. \nonumber
 \end{eqnarray}
 
 \begin{remark}
 \label{rmrecur}
 We can conclude that, when we are super-replicating backwards, the value process $\mathcal{V}_{k}(S_{k-H},S_{k})$ in (\ref{supereplvalueprocess22}) is only required at the two extreme points $S_{k}=S_{k-H}u^{H}$ and $S_{k}=S_{k-H}d^{H}$, because of the form on the right hand side of the recursive formula (\ref{recur}). In other words, we just use $(\mathfrak{p}_u,\mathfrak{q}_u)=(\mathfrak{p}_H,\mathfrak{q}_H)$ and $(\mathfrak{p}_d,\mathfrak{q}_d)=(\mathfrak{p}_0,\mathfrak{q}_0)$.
 \end{remark}
 \begin{figure}
 
 \centering
 \tikzstyle{bag} = [text width=2em, text centered]
 \tikzstyle{end} = []
 \begin{tikzpicture}[scale=0.8]
   \node (a) at ( 0,0) [bag] {$S_0$};
   \node (b) at ( 2,-0.75) [bag] {$S_0 d$};
   \node (c) at ( 2,0.75) [bag] {$S_0 u$};
   \node (d) at ( 4,-1.5) [bag] {$S_0 d^2$};
   \node (e) at ( 4,0) [bag] {$S_0 u d$};
   \node (f) at ( 4,1.5) [bag] {$S_0 u^2$};
    \node (g) at ( 6,-2.25) [bag] {$S_0 d^3$};
   \node (h) at ( 6,-0.75) [bag] {$S_0 u d^2$};
    \node (i) at ( 6,0.75) [bag] {$S_0 u^2 d$};
   \node (j) at ( 6,2.25) [bag] {$S_0 u^3$};
  \node (k) at ( 8,-3) [bag] {$S_0 d^4$};
    \node (l) at ( 8,-1.5) [bag] {$S_0 u d^3$};
   \node (m) at ( 8,0) [bag] {$S_0 u^2 d^2$};
    \node (n) at ( 8,1.5) [bag] {$S_0 u d^3$};
   \node (o) at ( 8,3) [bag] {$S_0 u^4$};
  \node (aaa) at (0,3){};
   \node (aa) at (0,-3.5){time 0};  
   \node (bb) at (2,-3.5){time 1};
   \node (cc) at (2,3){};
   \node (dd) at (4,-3.5){time 2};
   \node (ff) at (4,3){};
   \node (gg) at (6,-3.5){time 3};
   \node (jj) at (6,3){};
     \node (kk) at (8,-3.5){time 4};
   \node (oo) at (8,3){};
   \draw [->] (a) to node [below] {} (b);
   \draw [->] (a) to node [above] {} (c);
   \draw [->] (c) to node [below] {} (f);
   \draw [->] (c) to node [above] {} (e);
   \draw [->] (b) to node [below] {} (e);
   \draw [->] (b) to node [above] {} (d);
     \draw [->] (d) to node [below] {} (g);
   \draw [->] (d) to node [above] {} (h);
   \draw [->] (e) to node [below] {} (h);
     \draw [->] (e) to node [above] {} (i);
   \draw [->] (f) to node [below] {} (i);
   \draw [->] (f) to node [above] {} (j);
   \draw [->] (g) to node [below] {} (k);
 \draw [->] (g) to node [above] {} (l);  
  \draw [->] (h) to node [below] {} (l);
 \draw [->] (h) to node [above] {} (m);
  \draw [->] (i) to node [below] {} (m);
 \draw [->] (i) to node [above] {} (n);
  \draw [->] (j) to node [below] {} (n);
 \draw [->] (j) to node [above] {} (o);
 
 \draw [dashed] (aa) to node [above]{} (aaa);
   \draw [dashed] (cc) to node [above]{} (bb);
   \draw [dashed] (dd) to node [above]{} (ff);
 \draw [dashed] (gg) to node [above]{} (jj);
 \draw [dashed] (kk) to node [above]{} (oo);
 \end{tikzpicture}
   \caption{Asset price process $S_{k}$ in a 4-period binomial model}
   \label{fign4}
   \end{figure}
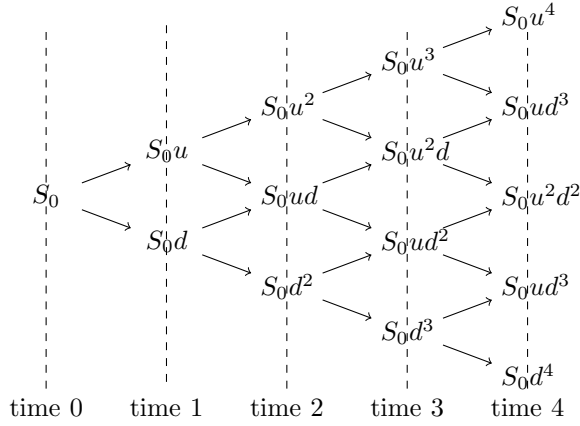
 Therefore, similar to (\ref{superrepeasy}), the super-replication price $\bar{\pi}(\varphi)$ can be finally calculated as
 \begin{eqnarray}
 \label{superrep}
 \bar{\pi}(\varphi)\,=\,e^{-rH} \max\left(\mathcal{V}_{H}(S_{0},S_{0}u^{H}),\mathcal{V}_{H}(S_{0},S_{0}d^{H})\right).
 \end{eqnarray}

 \subsubsection{Direct Approach}
 In this section, we solve the recursive equation (\ref{recur}) and obtain the value process $\mathcal{V}_{k}(S_{k-H},S_{k})$ for the super-replicating strategy explicitly. As Remark (\ref{rmrecur}) suggests, when we super-replicate backwards, we just need the value process at the extreme points, that is $\mathcal{V}_{k}(S_{k-H},S_{k-H}u^{H})$ and $\mathcal{V}_{k}(S_{k-H},S_{k-H}d^{H})$, $k=H,\dots,N-1$. \par
  Define probability spaces $(\Omega_{k},\mathcal{F}_{k},\mathbb{Q}_{k})$ for  $k=H,\dots,N-1$ with $\Omega_{k}=\{0,1\}^{\widetilde{N}+H}$, the Borel $\sigma$-algebra $\mathcal{F}_{k}$ on $\Omega_{k}$, and $\widetilde N = N-k$. For every $\omega_{k}=(\omega_{k,1},\dots,\omega_{k,\widetilde{N}+H})\in \Omega_{k}$, we define a coordinate map by $Z_{k,m}(\omega_{k})=\omega_{k,m}$ for each  $m \in \{1,\dots,\widetilde{N}+H\}$. \par
 Let $\mathbb{Q}_{k}$ be the probability measure under which $Z_{k,m}, m = 1, \ldots , \widetilde{N}+H$ with initial position $Z_{k,0}$ is a Markov chain, and for $ l=1,\dots,\widetilde{N}-1$, it has transition matrix
 \begin{eqnarray}
 \label{transition} 
 Q=
 \left( \begin{array}{ccc}
 \mathfrak{q}_{d} \, \quad \mathfrak{p}_{d}   \\
 \mathfrak{q}_{u} \, \quad \mathfrak{p}_{u}  \end{array} \right) \quad  \text{on} \, \, \{0,1\}.
 \end{eqnarray}
 Besides, for $l=\widetilde{N},\dots,\widetilde{N}+H$,
 \begin{eqnarray}
 \label{lastbigmove}
 \mathbb{Q}_{k}\left(Z_{k,\widetilde{N}+H}=\dots=Z_{k,\widetilde{N}}=1 \lvert Z_{k,\widetilde{N}-1}=1\right)\,=\,\mathfrak{p}_{u},  \,
 \mathbb{Q}_{k}\left(Z_{k,\widetilde{N}+H}=\dots=Z_{k,\widetilde{N}}=-1 \lvert  Z_{k,\widetilde{N}-1}=1\right)\,=\,\mathfrak{q}_{u}, \nonumber \\
 \mathbb{Q}_{k}\left(Z_{k,\widetilde{N}+H}=\dots=Z_{k,\widetilde{N}}=1 \lvert  Z_{k,\widetilde{N}-1}=0\right)\,=\,\mathfrak{p}_{d}, \,
 \mathbb{Q}_{k}\left(Z_{k,\widetilde{N}+H}=\dots=Z_{k,\widetilde{N}}=-1 \lvert Z_{k,\widetilde{N}-1}=0\right)\,=\,\mathfrak{q}_{d}.  \nonumber \\ 
 \end{eqnarray}
 The risky asset price $S_{k-H+m}$ satisfies  
 \begin{eqnarray}
 S_{k-H+m}:=S_{k-H}u^{I_{k,m}}d^{m-I_{k,m}} , \quad I_{k,m}=\sum_{l=1}^{m}Z_{k,l} , \quad m = 1, \ldots , \widetilde{N}+H \, .
 \end{eqnarray}
 \begin{remark}
 \label{moves}
 Under measures $\mathbb{Q}_{k}$, $k=H,\dots,N-1$, $\mathfrak{p}_{u}$ is the probability of an upward move preceded with an upward move, $\mathfrak{q}_{u}$ is the probability of a downward move preceded with an upward move, $\mathfrak{p}_{d}$ is the probability of an upward move preceded with a downward move, and $\mathfrak{q}_{d}$ is the probability of a downward move preceded with a downward move. Besides, equations (\ref{lastbigmove}) are to ensure that the last $H+1$ moves are all either upward or downward.
 \end{remark}
 \begin{remark}
 Under measures $\mathbb{Q}_{k}$, $k=H,\dots,N-1$, probability of a downward move
 preceded by a downward  move ($\mathfrak{q}_{d}$) is higher than the probability of a
 downward move preceded by an upward move ($\mathfrak{q}_{d}$). Similar is also true for upward moves. So, the variance of the risky asset price is
 higher under these measures than the initial measure $\mathbb{P}$.
 \end{remark}
 \begin{remark}
 \label{Hzero}
 If we put $H=0$, the transition matrix (\ref{transition}) would have duplicate rows (i.e. $\mathfrak{p}_{u}=\mathfrak{p}_{d}$ and $\mathfrak{q}_{u}=\mathfrak{q}_{d}$). Therefore, in this case, the model boils down to the binomial tree model of \cite{cox1979option}, and all the equations get significantly simplified accordingly.
 \end{remark}
 
 Theorem (\ref{valrecur}) expresses $\mathcal{V}_{k}(S_{k-H},S_{k-H}u^{H})$ and $\mathcal{V}_{k}(S_{k-H},S_{k-H}d^{H})$, $k=H,\dots,N-1$ as expectations under the measure $\mathbb{Q}_{k}$.
  \begin{theorem}
 \label{valrecur} For a European-style contingent claim with payoff $\, \varphi \, :=\, \Phi(S_{N}) \, $ for some convex function $\Phi(S_N) \in \mathbb{L}^{\infty}(\Omega_{k},\mathfrak{F}_{k},\mathbb{Q}_{k})$, $k=H,\dots,N-1$, the value process  $\mathcal{V}_k(S_{k-H},S_{k-H}u^{H})$ and $\mathcal{V}_k(S_{k-H},S_{k-H}d^{H})$, $k=H,\dots,N-1$ for the super-replicating strategy, in an $N$-period binomial model with $H$ periods of delay, can be calculated as
 \begin{eqnarray}
 \label{valup}
 \mathcal{V}_{k}(S_{k-H},S_{k-H}u^{H})\,=\,e^{-r\widetilde{N}}\mathbb{E}^{\mathbb{Q}_{k}}\left(\Phi\left(S_{N}\right) \lvert Z_{k,0}=1 \right), 
 \end{eqnarray}
 \begin{eqnarray}
 \label{valdown}
 \mathcal{V}_{k}(S_{k-H},S_{k-H}d^{H})\,=\,e^{-r\widetilde{N}}\mathbb{E}^{\mathbb{Q}_{k}}\left(\Phi\left(S_{N}\right) \lvert Z_{k,0}=0 \right).
 \end{eqnarray}
 \end{theorem}
 \begin{proof}
 We need to show that (\ref{valup}) and (\ref{valdown}) satisfy the recursive equation (\ref{recur}) for $k=H,\dots,N-2$, and equation (\ref{dp}) for $k=N-1$. For $k=N-1$, it is already shown in (\ref{rep}), and for $k=H,\dots,N-2$, by conditioning on $Z_{k,1}$, (\ref{valup}) satisfies
 \begin{eqnarray}
 \mathcal{V}_{k}(S_{k-H},S_{k-H}u^{H})&=&e^{-r\widetilde{N}}\mathbb{E}^{\mathbb{Q}_{k}}\left(\Phi\left(S_{N}\right) \lvert Z_{k,0}=1 \right),  \nonumber \\
                                  &=&e^{-r\widetilde{N}}\Bigg[\mathbb{E}^{\mathbb{Q}_{k}}\left(\Phi\left(S_{N}\right) \lvert Z_{k,0}=1, Z_{k,1}=0\right)\mathbb{Q}_{k}(Z_{k,1}=0 \lvert Z_{k,0}=1)  \nonumber \\
                                  & & + \mathbb{E}^{\mathbb{Q}_{k}}\left(\Phi\left(S_{N}\right) \lvert Z_{k,0}=1, Z_{k,1}=1\right)\mathbb{Q}_{k}(Z_{k,1}=1 \lvert Z_{k,0}=1)\Bigg]. \nonumber
 \end{eqnarray}
 
 Note that by the way the spaces $(\Omega_{k},\mathcal{F}_{k},\mathbb{Q}_{k})$ and $(\Omega_{k+1},\mathcal{F}_{k+1},\mathbb{Q}_{k+1})$ are constructed, 
 \begin{eqnarray}
 \mathbb{E}^{\mathbb{Q}_{k}}\left(\Phi\left(S_{N}\right) \lvert Z_{k,0}=1, Z_{k,1}=0\right)=e^{-r}\mathbb{E}^{\mathbb{Q}_{k+1}}\left(\Phi\left(S_{N}\right) \lvert Z_{k+1,0}=0\right), \nonumber
 \end{eqnarray}
 \begin{eqnarray}
 \mathbb{E}^{\mathbb{Q}_{k}}\left(\Phi\left(S_{N}\right) \lvert Z_{k,0}=1, Z_{k,1}=1\right)=e^{-r}\mathbb{E}^{\mathbb{Q}_{k+1}}\left(\Phi\left(S_{N}\right) \lvert Z_{k+1,0}=1\right). \nonumber
 \end{eqnarray}
 Also, $\mathbb{Q}_{k}(Z_{k,1}=0 \lvert Z_{k,0}=1)=\mathfrak{q}_{u}$ and $\mathbb{Q}_{k}(Z_{k,1}=1 \lvert Z_{k,0}=1)=\mathfrak{p}_{u}$. Therefore,
 \begin{eqnarray}
 \mathcal{V}_{k}(S_{k-H},S_{k-H}u^{H})&=&e^{-r\widetilde{N}}\left[\mathfrak{p}_{u}\mathbb{E}^{\mathbb{Q}_{k+1}}\left(\varphi\left(S_{N}\right) \lvert Z_{k+1,0}=1\right)+\mathfrak{q}_{u}\mathbb{E}^{\mathbb{Q}_{k+1}}\left(\varphi\left(S_{N}\right) \lvert Z_{k+1,0}=0\right) \right], \nonumber \\
 &=& e^{-r}\bigg [\mathfrak{p}_{u}\mathcal{V}_{k+1}(S_{k-H+1}=S_{k-H}u,S_{k+1}=S_{k-H}u^{H+1}) \nonumber \\
 & & + \mathfrak{q}_{u}\mathcal{V}_{k+1}(S_{k-H+1}=S_{k-H}d,S_{k+1}=S_{k-H}d^{H+1})\bigg ], \nonumber
 \end{eqnarray}
 which completes the proof. Similarly, it can also be shown for (\ref{valdown}).
 \end{proof}
 \begin{remark}
 \label{suprepexpl}
 If we are interested just to find out the time-$0$ super-replication price $\bar{\pi}(\varphi)$, we only need the probability space $(\Omega_{H},\mathcal{F}_{H},\mathbb{Q}_{H})$ where $\Omega_{H}=\{0,1\}^{N}$. Then, we would have
 \begin{eqnarray}
 \bar{\pi}(\varphi)=e^{-rN} \max\left\{\mathbb{E}^{\mathbb{Q}_{H}}\left(\Phi\left(S_{N}\right) \lvert Z_{H,0}=1 \right),\mathbb{E}^{\mathbb{Q}_{H}}\left(\Phi\left(S_{N}\right) \lvert Z_{H,0}=0 \right)\right\}.
 \end{eqnarray}
 \end{remark}
 
 Lemma \ref{lem1}, whose proof can be found in the appendix, calculates $\mathbb{E}^{\mathbb{Q}_{k}}\left(\Phi\left(S_{N}\right) \lvert Z_{k,0}=1\right)$, $k=H,\dots,N-1$. For $H+1 \leq i \leq \widetilde{N}+H-1, 1 \leq j \leq \min(i-H,\widetilde{N}+H-i)$. Define
 \begin{eqnarray}
 \label{fij}
  f(i,j)&:=&{\widetilde{N}+H-i-1 \choose j-1}{i-H \choose j} \mathfrak{q}_{u}^{(j)} \mathfrak{q}_{d}^ {(\widetilde{N}+H-i-j)} \mathfrak{p}_{u}^{(i-j-H)} \mathfrak{p}_{d}^{(j)}. 
 \end{eqnarray}
 Also for $0 \leq i \leq \widetilde{N}-2, 1 \leq j \leq \min(i+1,\widetilde{N}-i-1)$, define
 \begin{eqnarray}
 \label{hij}
  h(i,j) &:=&{\widetilde{N}-i-2 \choose j-1}{i \choose j-1} \mathfrak{q}_{u}^{(j)} \mathfrak{q}_{d}^{(\widetilde{N}-i-j)} \mathfrak{p}_{u}^{(i-j+1)} \mathfrak{p}_{d}^{(j-1)} \nonumber \\
 	& & + {\widetilde{N}-i-2 \choose j-1}\left[{i+1 \choose j}-{i \choose j-1}\right] \mathfrak{q}_{u}^{(j+1)} \mathfrak{q}_{d}^{(\widetilde{N}-i-j-1)} \mathfrak{p}_{u}^{(i-j)} \mathfrak{p}_{d}^{(j)} . 
 \end{eqnarray}
  
 \begin{lemma}
 \label{lem1}
 For a function $\Phi(S_{N}) \in \mathbb{L}^{\infty}(\Omega_{k},\mathfrak{F}_{k},\mathbb{Q}_{k})$, $k=H,\dots,N-1$, the conditional expectation $\mathbb{E}^{\mathbb{Q}_{k}}\left(\Phi\left(S_{N}\right) \lvert Z_{k,0}=1 \right)$ can be explicitly calculated as 
 \begin{eqnarray}
 \mathbb{E}^{\mathbb{Q}_{k}}\left(\Phi\left(S_{N}\right) \lvert Z_{k,0}=1 \right) = \sum\limits_{i=0}^{\widetilde{N}+H} \mathbb{Q}_{k}\left(S_N=S_{k-H} u^i d^{\widetilde{N}+H-i} \lvert Z_{k,0}=1\right)\Phi(S_{k-H} u^i d^{\widetilde{N}+H-i}) , 
 \end{eqnarray}
 where  $\,\mathbb{Q}_{k}\left(S_N=S_{k-H} u^i d^{\widetilde{N}+H-i}\,  \lvert \, Z_{k,0}=1\right) \,$ is given by 
 \begin{eqnarray}
 \label{measurexx}
 \left \{ 
 \begin{array}{ll}
 \sum\limits_{j=1}^{\min(i+1,\widetilde{N}-i-1)} h(i,j) \,& 0 \leq i \leq H;   \\
 \\
 \sum\limits_{j=1}^{\min(i+1,\widetilde{N}-i-1)} h(i,j)+\sum\limits_{j=1}^{\min(i-H,\widetilde{N}+H-i)} f(i,j) \,& H+1 \leq i \leq \widetilde{N}-2; \\
 \\
 \mathfrak{p}_{u}^{(\widetilde{N}-1)} \mathfrak{q}_{u} + \sum\limits_{j=1}^{\min(\widetilde{N}-H -1,H+1)} f(i,j) \,& i=\widetilde{N}-1;   \\
 \\
 \sum\limits_{j=1}^{\min(i-H,\widetilde{N}+H-i)} f(i,j) \,& \widetilde{N} \leq i \leq \widetilde{N}+H-1;  \\
 \\
 \mathfrak{p}_{u}^{(\widetilde{N})}	\,& i=\widetilde{N}+H. 
 \end{array}
 \right . 
 \end{eqnarray} 
 \end{lemma}
 
 Similarly, Lemma \ref{lem2} calculates $\mathbb{E}^{\mathbb{Q}_{k}}\left(\Phi\left(S_{N}\right) \lvert Z_{k,0}=0\right)$, $k=H,\dots,N-1$. 
 Also for $H+2 \leq i \leq \widetilde{N}+H, 1 \leq j \leq \min(i-H-1,\widetilde{N}+H-i+1)$, define
 \begin{eqnarray}
  \widetilde{f}(i,j) &:=&{i-H-2 \choose j-1}{\widetilde{N}+H-i \choose j-1} \mathfrak{q}_{u}^{(j-1)} \mathfrak{q}_{d}^{(\widetilde{N}+H-i-j+1)} \mathfrak{p}_{u}^{(i-j-H)} \mathfrak{p}_{d}^{(j)} \nonumber \\
 	&+& {i-H-2 \choose j-1}\left[{\widetilde{N}+H-i+1 \choose j}-{\widetilde{N}+H-i \choose j-1}\right] \mathfrak{q}_{u}^{(j)} \mathfrak{q}_{d}^{(\widetilde{N}+H-i-j)} \mathfrak{p}_{u}^{(i-j-H-1)} \mathfrak{p}_{d}^{(j+1)}.
 \end{eqnarray}
 
 For $1 \leq i \leq \widetilde{N}-1, 1 \leq j \leq \min(i,\widetilde{N}-i)$, define 
 \begin{eqnarray}
  \widetilde{h}(i,j)&:=&{i-1 \choose j-1}{\widetilde{N}-i \choose j} \mathfrak{q}_{u}^{(j)} \mathfrak{q}_{d}^ {(\widetilde{N}-i-j)} \mathfrak{p}_{u}^{(i-j)} \mathfrak{p}_{d}^{(j)}.  
 \end{eqnarray}
%
 
 \begin{lemma}
 \label{lem2}
 For a function $\Phi(S_{N}) \in \mathbb{L}^{\infty}(\Omega_{k},\mathfrak{F}_{k},\mathbb{Q}_{k})$, $k=H,\dots,N-1$, the conditional expectation $\mathbb{E}^{\mathbb{Q}_{k}}\left(\Phi\left(S_{N}\right) \lvert Z_{k,0}=1 \right)$ can be explicitly calculated as 
 \begin{eqnarray}
 \mathbb{E}^{\mathbb{Q}_{k}}\left(\Phi\left(S_{N}\right) \lvert Z_{k,0}=1 \right) = \sum\limits_{i=0}^{\widetilde{N}+H} \mathbb{Q}_{k}\left(S_N=S_{k-H} u^i d^{\widetilde{N}+H-i} \lvert Z_{k,0}=1\right)\Phi(S_{k-H} u^i d^{\widetilde{N}+H-i}) , 
 \end{eqnarray}
 where  $\,\mathbb{Q}_{k}\left(S_N=S_{k-H} u^i d^{\widetilde{N}+H-i} \, \lvert \, Z_{k,0}=1\right)\,$ is given by 
 
 \begin{eqnarray}
 \left \{ 
 \begin{array}{ll}
      \mathfrak{q}_{d}^{(\widetilde{N})}                               \,& i=0  \\
 \\
 \sum\limits_{j=1}^{\min(i,\widetilde{N}-i)} \widetilde{h}(i,j) \,& 1 \leq i \leq H;  \\
 \\
 \mathfrak{q}_{d}^{(\widetilde{N}-1)} \mathfrak{p}_{d} + \sum\limits_{j=1}^{  \min(H+1,\widetilde{N}-H-1) } \widetilde{h}(i,j) \,& i=H+1        ;  \\
 \\
 \sum\limits_{j=1}^{    \min(i,\widetilde{N}-i)                   } \widetilde{h}(i,j)+\sum\limits_{j=1}^{\min(i-H-1,\widetilde{N}+H-i+1)} \widetilde{f}(i,j) \,& H+2 \leq i \leq \widetilde{N}-1;           \\
 \\
 \sum\limits_{j=1}^{ \min(i-H-1,\widetilde{N}+H-i+1) } \widetilde{f}(i,j) \,& \widetilde{N} \leq i \leq \widetilde{N}+H  . \\     
 \end{array}
 \right . 
 \end{eqnarray} 
 \end{lemma}
 \begin{proof}
 The proof follows very similarly as that of Lemma \ref{lem1} with only this difference that since $Z_{k,0}=0$, we look for upward groups instead of downward groups.
 \end{proof}
 
 \subsection{Geometrical Representation}
 \label{geometry}
 In this subsection, we first discuss Theorem \ref{thm1} from a geometrical perspective Then, we represent the dynamic programming approach in subsection \ref{dynamic} geometrically For convenience, assume that interest rate $r=0$, and $H=1$ in this subsection. \par
 In Theorem \ref{thm1}, we discussed that in an $N$-period binomial model with $H=N-1$ periods of delay, for a European-style convex contingent claim with payoff function $\varphi:=\Phi(S_{N}) \in \mathbb{L}^{\infty}(\Omega,\mathfrak{F},\mathbb{P})$, there exist $\Delta_H ^{\ast}$ and $x_{0}^{\ast}$ such that
 \begin{eqnarray}
 \mathcal{V}_{H}(S_{0},S_{H})\,=\,x_{0}^{\ast}+\Delta_{H}^{\ast}S_{H}.
 \end{eqnarray}
 This suggests that there exists a line with slope $\Delta_{H}^{\ast}$ and intercept $x_{0}^{\ast}$ such that the super-replicating value function $\mathcal{V}_{H}(S_{0},S_{H})$ lie on that line. Figure \ref{figg2} shows this optimal line, the super-replication price, and the super-replicating value functions in a $2$-period binomial model with $1$ period of delay.
 \begin{figure}
 
 \centering
 \begin{tikzpicture}[scale=1.5]
   \draw[->] (-0.2,0) -- (4,0) node[right] {$S_2$};
   \draw[->] (0,-0.2) -- (0,2.25) node[above] {$\Phi(S_2)$};
 \foreach \x/\xtext in {0.7/\text{\scriptsize $S_0 d^2$}, 1.25/\text{\scriptsize $S_0 d$}, 1.65/\text{\scriptsize $S_0$}, 2/\text{\scriptsize $S_0 ud$}, 2.5/\text{\scriptsize $S_0 u$}, 3/\text{\scriptsize $S_0 u^2$}}
     \draw[shift={(\x,0)}] (0pt,2pt) -- (0pt,-2pt) node[below] {$\xtext$};
 \node (b) at (0,0.27) [left] {$x_{0}^{\ast}$};
 \node (ff) at (2.5,1.56){};
   \draw (0.25,1) parabola bend (1.5,0.5) (3.25,2.25) node[right] {$\Phi$};
  \draw [dashed] (2.5,0) coordinate (a_1) -- (ff) ;
 \draw [dashed] (0.7,0) -- (0.7,0.7);
 \fill (0.7,0.7) circle (1pt);
 \draw [dashed] (3,0) coordinate (b_1) -- (3,1.8) coordinate (b_2);
 \draw [dashed] (1.25,0) coordinate (x_1) -- (1.25,0.96) coordinate (x_2);
 \draw [dashed] (x_2) -- (0,0.96) coordinate (x_3);
 \draw [dashed] (3,0) coordinate (b_1) -- (3,1.8) coordinate (b_2);
 \draw (b) -- (3.5,2.03) node [right]{\text{Optimal Line}};
 \fill (ff) circle (1pt);
 \fill (b_2) circle (1pt);
 \fill (x_2) circle (1pt);
 \fill (2,0.65) circle (1pt);
 \node (c) at (0,1.56) [left] {$\mathcal{V}_{1}(S_{0},S_{1}=S_0u)$};
 \node (d) at (0,0.96) [left] {$\mathcal{V}_{1}(S_{0},S_{1}=S_0d)$};
 \draw [dashed] (ff) -- (c);
 \end{tikzpicture}
 \caption{Super-replicating Strategy in a $2$-period binomial Model with a $1$-period Delay. {The optimal line characterizes the super-replicating strategy. The slope of it is $\Delta_1 ^{\ast}$ and its intercept is  $x_{0}^{\ast}$}. The super-replication price is $\bar{\pi}(\varphi)=\max\left\{\mathcal{V}_{1}(S_{0},S_{1}=S_0d),\mathcal{V}_{1}(S_{0},S_{1}=S_0u)\right\}$.
 }
 \label{figg2}
 \
 \end{figure}
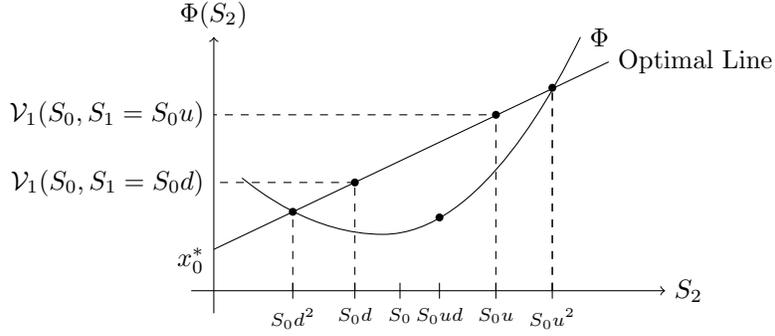
 
 It is more intuitive to demonstrate the dynamic programming approach in subsection \ref{dynamic} geometrically. Figure \ref{fign4} shows a $4$-period binomial model with $1$-period delay. Figure \ref{ffgg} shows how to 
 geometrically find the super-replication price for a contingent claim with convex payoff function ($\Phi(.)$). For convenience and to avoid a clutter of points on the $x$-axis, suppose $u d =1$, so some of the points in the model lie on each other.\par 
 
  Now in order to find the super-replication prices at time $3$, it is necessary to consider the three $2$-period binomial models with $1$-period delay $\mathcal T_{2} (2,0)$, $\mathcal T_{2} (1,1)$ and $\mathcal T_{2} (0,2)$. In Figure \ref{ffgg}, the lines ($om$), ($nl$) and ($mk$) show the optimal super-replication lines for each of these models
 respectively. As it can be seen, there are two payoffs at either of the nodes $S_{3}=S_{0}u^{2}d$ and $S_{3}=S_{0}ud^{2}$  depending on which subtree is used for pricing (i.e. depending on what $S_{1}$ is). Now, we go one period further back to find out 
the payoffs at time $2$. We need to consider two $2$-period models $\mathcal T_{2} (1,0)$ and $\mathcal T_{2} (0,1)$. Note that in each of these two models, the corresponding payoff at nodes $S_{2}=S_{0}u^{2}d$ (out of two payoffs $\Phi_{\mathcal T_{2} \left(1,0\right)}\left(2,1\right) \,$ and $\Phi_{\mathcal T_{2} \left(0,1\right)}\left(2,1\right) \,$) and $S_{2}=S_{0}ud^{2}$ (out of two payoffs $\Phi_{\mathcal T_{2} \left(1,0\right)}\left(1,2\right) \,$ and $\Phi_{\mathcal T_{2} \left(1,0\right)}\left(1,2\right) \,$) needs to be chosen. As Theorem (\ref{thm2}) suggests,  the payoff functions for both of these models are convex. The lines ($jh$) and ($ig$) demonstrate the optimal lines for these models.  Similarly, to calculate the payoff at time $1$, the $2$-period model $\mathcal T_{2} (0,0)$ needs to be used and the line ($fd$) shows the optimal line for this model. Finally, we have the super-replication price $\bar{\pi}(\varphi)=\max\left\{\mathcal{V}_{1}(S_{0},S_{1}=S_0d),\mathcal{V}_{1}(S_{0},S_{1}=S_0u)\right\}$. \par

    \begin{figure}
 
    \centering
 \begin{tikzpicture}[scale=2.3]
 
   \draw[->] (-0.2,0) -- (3.5,0) node[right] {$S_4$};
   \draw[->] (0,-0.2) -- (0,2.25) node[above] {$\Phi(S_4)$};
 
 \foreach \x/\xtext in {0.25/\text{\tiny $S_0 d^4$},0.6/\text{\tiny $S_0 d^3$}, 1/\text{\tiny $S_0 u d^3$},1.4/\text{\tiny $S_0 u d^2$},
  1.75/\text{\tiny $S_0 u^2 d^2$},2.1/\text{\tiny $S_0 u^2 d$}, 2.5/\text{\tiny $S_0 u^3 d$},2.9/\text{\tiny $S_0 u^3$}, 3.25/\text{\tiny $S_0 u^4$}}
     \draw[shift={(\x,0)}] (0pt,2pt) -- (0pt,-2pt) node[below] {$\xtext$};

 \foreach \x/\xtext in {0.25/\text{},0.6/\text{}, 1/\text{\tiny $S_0 d^2$},1.4/\text{\tiny $S_0 d$},
  1.75/\text{\tiny $S_0 u d$},2.1/\text{\tiny $S_0 u $}, 2.5/\text{\tiny $S_0 u^2 $},2.9/\text{}, 3.25/\text{}}
     \draw[shift={(\x,-0.2)}] (0pt,0pt) -- (0pt,0pt) node[below] {$\xtext$};
 
 \foreach \x/\xtext in {0.25/\text{},0.6/\text{}, 1/\text{},1.4/\text{},
  1.75/\text{\tiny $S_0$},2.1/\text{}, 2.5/\text{},2.9/\text{}, 3.25/\text{}}
     \draw[shift={(\x,-0.3)}] (0pt,0pt) -- (0pt,0pt) node[below] {$\xtext$};

   \draw (0.1,1) parabola bend (1.8,0.5) (3.4,2.1) node[right]{$\Phi$};
 
 \node (k) at (0.25,0.91){};
 \fill (k) circle (0.5pt);
 \draw [dashed,very thin] (k) -- (0.25,0);
 
 \node (g) at (0.6,0.75){};
 \fill (g) circle (0.5pt);
 \draw [dashed,very thin] (g) -- (0.6,0);
 
 \node (l) at (1,0.61){};
 \fill (l) circle (0.5pt);
 \draw [dashed,very thin] (l) -- (1,0);
 
 \node (h) at (1.4,0.53){};
 \fill (h) circle (0.5pt);
 \draw [dashed,very thin] (h) -- (1.4,0);
 
 \node (m) at (1.75,0.5){};
 \fill (m) circle (0.5pt);
 \draw [dashed,very thin] (m) -- (1.75,0);
 
 \node (i) at (2.1,0.56){};
 \fill (i) circle (0.5pt);
 \draw [dashed,very thin] (i) -- (2.1,0);
 
 \node (n) at (2.5,0.815){};
 \fill (n) circle (0.5pt);
 \draw [dashed,very thin] (n) -- (2.5,0);
 
 \node (j) at (2.9,1.25){};
 \fill (j) circle (0.5pt);
 \draw [dashed,very thin] (j) -- (2.9,0);
 
 \node (o) at (3.25,1.82){};
 \fill (o) circle (0.5pt);
 \draw [dashed,very thin] (o) -- (3.25,0);

 \draw [red](0.25,0.91) -- node[pos=0.5,sloped, above=-0.3em]{\tiny mk} (1.75,0.5) ;
 \draw [red](2.5,0.815) -- node[pos=0.5,sloped, above=-0.3em]{\tiny nl} (1,0.61) ;
 \draw [red](3.25,1.82) -- node[pos=0.5,sloped, above=-0.3em]{\tiny om} (1.75,0.5) ;
 
 \node (jj) at (2.9,1.52){};
 \draw [dashed,very thin] (j) -- (jj);
 \fill (jj) circle (0.5pt);
 
 \node (hh) at (1.4,0.67){};
 \draw [dashed,very thin] (h) -- (hh);
 \fill (hh) circle (0.5pt);
 
 \draw [blue,densely dashed,semithick](2.9,1.52) -- node[pos=0.5,sloped, above=-0.3em]{\tiny jh} (1.4,0.67) ;
 
 \node (ii) at (2.1,0.755){};
 \draw [dashed,very thin] (i) -- (ii);
 \fill (ii) circle (0.5pt);
 
 \node (gg) at (0.6,0.82){};
 \draw [dashed,very thin] (g) -- (gg);
 \fill (gg) circle (0.5pt);
 
 \draw [blue,densely dashed,semithick](2.1,0.755) -- node[pos=0.5,sloped, above=-0.3em]{\tiny ig} (0.6,0.82) ;
 
 \node (nnn) at (2.5,1.29){};
 \draw [dashed,very thin] (n) -- (nnn);
 \fill (nnn) circle (0.5pt);
 
 \node (lll) at (1,0.805){};
 \draw [dashed,very thin] (l) -- (lll);
 \fill (lll) circle (0.5pt);
 
 \draw [magenta,dotted,thick](2.5,1.29) -- node[pos=0.5,sloped, above=-0.3em]{\tiny fd} (1,0.805) ;
 
 \node (iii) at (2.1,1.17){};
 \draw [dashed,very thin] (ii) -- (iii);
 \fill (iii) circle (0.5pt);
 
 \node (v) at (0, 1.17)[left]{$\mathcal{V}_{1}(S_{0},S_{1}=S_0u)$};
 \node (v2) at (0, 0.96)[left]{$\mathcal{V}_{1}(S_{0},S_{1}=S_0d)$};
 \draw [dashed,very thin] (iii) -- (v);
 \draw [dashed,very thin] (1.4,0.67) -- (1.4,0.94) coordinate (xx) ;
 \fill (xx) circle (0.5pt);
 \draw [dashed,very thin] (xx) -- (0,0.94) ;

 \end{tikzpicture}
 \caption{Geometrical Representation of the Super replicating Strategy in a $4$-period binomial Model with $1$-period Delay using a Dynamic Programming Approach}
 \label{ffgg}
 
 \end{figure}

 \section{Continuous Time Model}
 \label{continuous}
 In this section, we discuss the asymptotic behavior of the model. We define the probability spaces $(\Omega^{n},\mathfrak{F}^{n},\mathbb{Q}^{n})$, $n \in \mathbb{N}$ such that $\Omega^{n}=\{0,1\}^{n}$, and $\mathfrak{F}^{n}$ is the Borel $\sigma$-algebra on $\Omega^{n}$. For every $\omega^{n}=(\omega^{n}_{1},\dots,\omega^{n}_{n})\in \Omega^{n}$, we define a coordinate map by $Z^{n}_{\ell}(\omega^{n})=\omega^{n}_{\ell}$ for each  $\ell \in \{1,\dots,n\}$. Define the filtration $ \{\mathcal{F}^{n}_{\ell}, \ell = 0, \ldots , n\}$, where $\mathcal{F}^{n}_{\ell}$ is the  $\sigma$-field $\sigma(Z^{n}_{1},\dots,Z^{n}_{\ell})$ generated by the first $\ell$ variables for $\ell=1,\dots,n$ and $\mathcal{F}_{0}$ is the trivial $\sigma$-field. \par
 Let $\mu, \sigma, r \in [0,\infty)$, $H \in \mathbb{N}$, $T>0$ (fixed time horizon) and define the sequences
 \begin{eqnarray}
 &&\mu_{n}\,=\,\mu T \delta_{n}^{2}, \quad
 \sigma_{n}\,=\,\sigma \sqrt{T}\delta_{n}, \quad
 u_{n}\,=\,\exp{(\mu_{n}+\sigma_{n})},  \nonumber\\
 && d_{n}\,=\,\exp{(\mu_{n}-\sigma_{n})}, \quad
 r_{n}\,=\,r T \delta_{n}^{2}, \quad
 H_{n}\,=\, H T \delta_{n}^{2}, 
 \end{eqnarray}
 where the order $\delta_{n}\,=\,\nicefrac{1}{\sqrt{n}}$, as in Donsker's theorem. \par
 \begin{remark}
 $H$ characterizes the number of periods we have delayed information, which is constant in the asymptotic analysis. However, $H_{n}$ is the amount of time we have delayed information, which should vanish in the limit. Otherwise, the super-replication price would explode and converge to the maximum of the contingent claim payoff function. 
 \end{remark}
 \subsection{Price Process Asymptotic}
 Define the probability measures $\mathbb{Q}^{n}$, similar to (\ref{transition}) and (\ref{lastbigmove}), such that $Z^{n}_{\ell}, \ell = 1, \ldots , n$ with initial position $Z^{n}_{0}$ is a Markov chain, and for $ \ell=1,\dots,n-H-1$, it has transition matrix
 \begin{eqnarray}
 \label{transitionb} 
 Q_{n}=
 \left( \begin{array}{ccc}
 \mathfrak{q}_{n,d} \, \quad \mathfrak{p}_{n,d}   \\
 \mathfrak{q}_{n,u} \, \quad \mathfrak{p}_{n,u}  \end{array} \right) \quad  \text{on} \, \, \{0,1\}.
 \end{eqnarray}
 Besides, for $m=n-H,\dots,n$,
 \begin{eqnarray}
 \label{lastbigmoveb}
 \mathbb{Q}^{n}\left(Z^{n}_{n}=\dots=Z^{n}_{n-H}=1 \lvert Z^{n}_{n-H-1}=1\right)\,=\,\mathfrak{p}_{n,u}, \,
 \mathbb{Q}^{n}\left(Z^{n}_{n}=\dots=Z^{n}_{n-H}=-1 \lvert Z^{n}_{n-H-1}=1\right)\,=\,\mathfrak{q}_{n,u},  \nonumber  \\
 \mathbb{Q}^{n}\left(Z^{n}_{n}=\dots=Z^{n}_{n-H}=1 \lvert Z^{n}_{n-H-1}=0\right)\,=\,\mathfrak{p}_{n,d},  \,
 \mathbb{Q}^{n}\left(Z^{n}_{n}=\dots=Z^{n}_{n-H}=-1 \lvert Z^{n}_{n-H-1}=0\right)\,=\,\mathfrak{q}_{n,d}, \nonumber \\
 \end{eqnarray}
 where $\mathfrak{p}_{n,u}, \mathfrak{q}_{n,u}, \mathfrak{p}_{n,d}$ and $\mathfrak{q}_{n,d}$ are defined, similar to (\ref{measure}) with $j=0, H$, as
 \begin{eqnarray}
 \label{pnu}
 \mathfrak{p}_{n,d}:=\frac{d_{n}^{H}e^{r_{n}}-d_{n}^{H+1}}{u_{n}^{H+1}-d_{n}^{H+1}}=1-\mathfrak{q}_{n,d}, \hspace{3mm}
 \mathfrak{p}_{n,u}:=\frac{u_{n}^{H}e^{r_{n}}-d_{n}^{H+1}}{u_{n}^{H+1}-d_{n}^{H+1}}=1-\mathfrak{q}_{n,u}. \quad
 \end{eqnarray}
 Then, the risky asset price $S^{n}_{\ell}$, similar to (\ref{priceprocess}), satisfies  
 \begin{eqnarray}
 S_{\ell}^{n}\,=\,S_{0}\exp{\left[\ell \mu_{n} + \sigma_{n}\sum\limits_{i=1}^{\ell}X^{n}_{i}\right]},  \quad \ell=0,\dots,n , 
 \end{eqnarray}
 where $X_{i}^{n}=2Z_{i}^{n}-1$. The following Lemma \ref{asymp} provides asymptotic for $\mathfrak{p}_{n,u}$ and $\mathfrak{p}_{n,d}$. 
 \begin{lemma}
 \label{asymp}
 We have
 \begin{eqnarray}
 \mathfrak{p}_{n,u}=\frac{2H+1}{2(H+1)}-\left(\frac{\mu-r}{2\left(H+1\right)\sigma}+\frac{2H+1}{4(H+1)}\sigma \right)\sqrt{T}\delta_{n}+\mathcal{O}\left(\delta_{n}^{2}\right), \\
 \mathfrak{p}_{n,d}=\frac{1}{2(H+1)}-\left(\frac{\mu-r}{2\left(H+1\right)\sigma}+\frac{2H+1}{4(H+1)}\sigma \right)\sqrt{T}\delta_{n}+\mathcal{O}\left(\delta_{n}^{2}\right).
 \end{eqnarray}
 \end{lemma}
 \begin{proof}
 The proof simply follows by applying Taylor's expansion to $u_{n}$, $d_{n}$ and $r_{n}$, and plugging them in (\ref{pnu}). 
 \end{proof}
 Discretize the time interval by setting $t_{\ell}^{n}:=T\ell/n$. By interpolating over the intervals $[t_{\ell-1}^{n},t_{\ell}^{n})$ in a piecewise constant manner with ($S_{\ell}^{n}$, $\ell=0, \ldots , n$), we get the risky asset price process $S^{(n)}=(S^{(n)}_{t})_{0 \leq t \leq T}$ 
 \begin{eqnarray}
 S^{(n)}_{t}:=S^{n}_{\lfloor nt \rfloor/T}, \quad 0 \leq t \leq T, 
 \end{eqnarray}
 where $\lfloor . \rfloor$ is the floor function. \par
 The process $S^{(n)}$ has trajectories which are right continuous with left limits. Note that in particular
 \begin{eqnarray}
 S^{(n)}_{t_{\ell}^{n}}=S^{n}_{\ell}, \quad \ell=0,\dots,n.  \nonumber
 \end{eqnarray}
 Here, $S^{(n)}$ under measure $\mathbb{Q}^{n}$ is distributed according to a probability measure $\rho_{n}$ on the Skorokhod space $\mathbb{D}[0,T]$ of right continuous functions with left limits. Theorem \ref{contlimit} provides a weak convergence for the sequence $(\rho_{n})_{n \in \mathbb{N}}$. 
 \begin{theorem}
 \label{contlimit}
 The sequence of processes $(S^{(n)})_{n \in \mathbb{N}}$ converges in distribution to the process $(S_{t})_{0 \leq t \leq T}$ with dynamics
 \begin{eqnarray}
 \label{st}
 d S_{t}=r S_{t} d t + \widetilde{\sigma} S_{t} d W_{t}, \quad 0 \leq t \leq T, 
 \end{eqnarray}
 where $(W_{t})_{0 \leq t \leq T}$ is a Brownian motion, and we have the enlarged volatility 
 \begin{eqnarray}
 \label{enlargedvol}
 \widetilde{\sigma}=\sqrt{2H+1}\sigma.
 \end{eqnarray}
 \end{theorem} 
 \begin{proof}
 First, note that
 \begin{eqnarray}
 \mathbb{Q}^{n}\left(X^{n}_{\ell}=1 \lvert  X^{n}_{\ell-1}\right)=\frac{1}{2}\left[\mathfrak{p}_{n,u}+\mathfrak{p}_{n,d}+X^{n}_{\ell-1}\left(\mathfrak{p}_{n,u}-\mathfrak{p}_{n,d}\right) \right], \quad \ell=1,\dots,n-H-1, \nonumber
 \end{eqnarray}
 According to Lemma \ref{asymp}, we conclude that
 \begin{eqnarray}
 \mathbb{Q}^{n}\left(X^{n}_{\ell}=1 \lvert  \mathcal{F}^{n}_{\ell-1}\right)=\mathfrak{p}_{\ell}\left(\ell,X^{n}_{\ell-1} \right), \quad \ell=1,\dots,n-H-1, \nonumber
 \end{eqnarray}
 where in the notation of \cite{gruber2006diffusion}
 \begin{eqnarray} \label{gaplambda}
 \mathfrak{p}_{n}(\ell,x)=\frac{1}{2}\left[1+\phi\delta_{n}+\lambda_{n}x \right]+\mathcal{O}(\delta_{n}^{2}), \,
 \phi=-2\left[\frac{\mu-r}{2\left(H+1\right)\sigma}+\frac{2H+1}{4\left(H+1\right)}\sigma\right]\sqrt{T}, \,
 \lambda_{n}=\frac{H}{H+1}+\mathcal{O}(\delta_{n}^{2}).  \nonumber \\
 \end{eqnarray}

Now  we apply a functional central limit theorem for generalized, correlated random walks in \cite{gruber2006diffusion} with $a_{n}(t,y):=\lambda_{n}$ and $b_n(t,y):=\phi$ for their Theorem 1 and Remark 3. It follows from the continuous mapping theorem  that $S^{(n)}$, regardless of the initial distribution of $X_{0}^{n}$, converges in distribution to $(S_{t})_{0 \leq t \leq T}$ in (\ref{st}). In particular, since $ \lim_{n\to \infty} \lambda_{n} \, =\, H/(H+1) $, we see the volatility 
 \begin{eqnarray}
 \widetilde{\sigma}= \sqrt{\frac{1+ \lim_{n\to \infty}a_{n}(t,Y_{t})}{1- \lim_{n\to \infty} a_{n}(t,Y_{t})}} \cdot \sigma = \sqrt{2H+1}\,  \sigma, \nonumber
 \end{eqnarray} 
is constant, and larger than $ \sigma $, where $Y_{t}=\log{S_{t}}$.
 
 \end{proof}
 
 \begin{remark}
 The enlarged volatility in the limit is due to the gap $\lambda_{n}=\mathfrak{p}_{n,u}-\mathfrak{p}_{n,d}$ in (\ref{gaplambda}) which is caused because of the delay in the flow of information (it would be zero when the number of delayed periods $H=0$). In fact, this is the main source making the price process under the pricing measure more volatile.
 \end{remark}

 \subsection{Exaggerated Volatility Smile}
 \label{smile}
 In this subsection, we discuss the volatility smile of the model, and how it evolves with the number of periods ($n$). Volatility smile is the graph of Black-Scholes implied volatility with respect to the strike price. Implied volatility is the value of the volatility in the Black-Scholes pricing model which generates a price equal to that of our model. Several market features, such as crashphobia, have been attributed as the culprits of the market smile. The volatility smile has been one of the central topics in option pricing literature, and many models have been developed to capture it. We refer to \cite{gatheral2011volatility} for more discussion in this regard. \par 
 
 Our model with delayed information shows that delayed information exaggerates the smile. Figure \ref{smile1} plots the volatility smiles for call and put options in the model with and without delayed information when $n=100$. In the model with delayed information ($H_{n}=\frac{1}{100} \text{ year}\approx 2.52 \text{ days}$), we observe volatility smile, on the contrary with the model without delayed information where we get an almost flat smile, which is expected according to the Remark \ref{Hzero}. Note that in the model with delayed information, we have different smiles for call and put option, and that is because  there is not any call-put parity, as discussed in Remark \ref{callputparity}. \par 
 
 Figure \ref{smile1} plots the volatility smiles for call and put options when the number of periods is very big ($n=250,000$) for the model with delayed information ($H_{n}=\frac{1}{250,000} \text{ year}\approx 30\text{ seconds}$). We observe almost the same flat volatility smiles for both call and put options, which can be also calculated by the theoretical results in \ref{enlargedvol}. \par 
 
 These volatility smiles in Figures \ref{smile1} and \ref{smile2} confirm the intuition of traders that delayed information would exaggerate the volatility smile, but it is not its culprit. This is because in the continuous limit, volatility is constant and there is no smile, but in the discrete model, we can observe volatility smile. Therefore, it conveys that the smile observed in the market might have been exaggerated by the way we interact with delayed information, and the smile might not be caused all by the market itself.

 \begin{figure}
 \centering
 \includegraphics[scale=0.75]{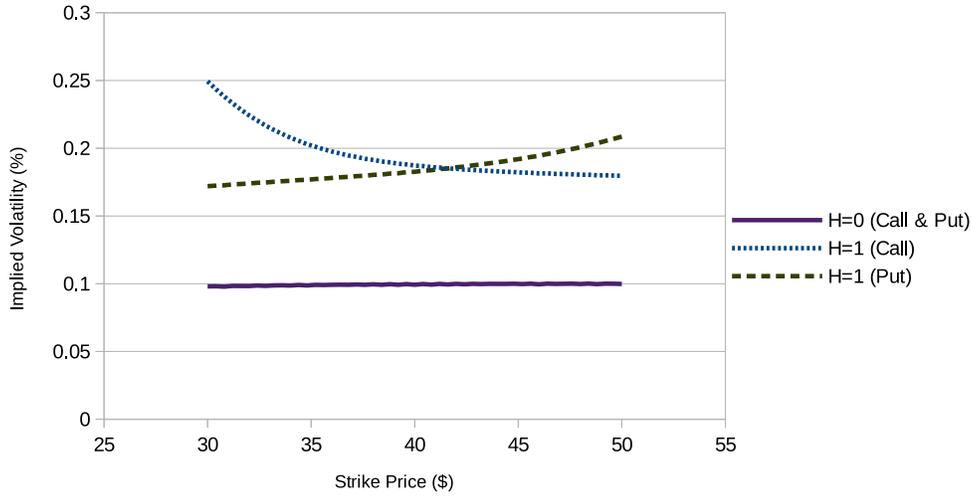}
 \caption{Volatility smile for the Call and Put options in the binomial model with and without delayed information ($H_{n}=\frac{1}{100} \text{ year}\approx 2.52 \text{ days}$ and $0 \text{ day}$ respectively). The parameters are $\sigma=0.1$, $T=1$, $r=0$, $S_0=40$, and $n=100$}
 \label{smile1}
 \end{figure}

  \begin{figure}
  
   \centering
   \includegraphics[scale=0.75]{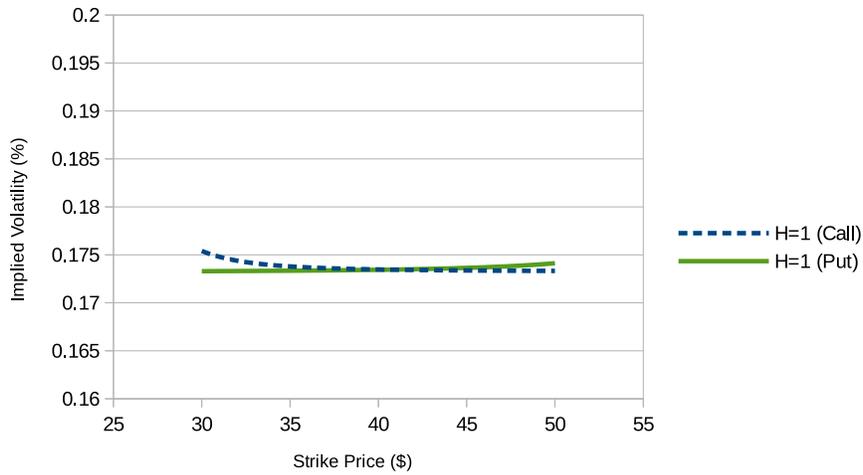}
   \caption{Volatility smile for the Call and Put options in the binomial model with delayed information ($H_{n}=\frac{1}{250,000} \text{ year}\approx 30\text{ seconds}$). The parameters are $\sigma=0.1$, $T=1$, $r=0$, $S_0=40$, and $n=250,000$}
   \label{smile2}
  \end{figure}

  \appendix
  \section{Proof of Lemma \ref{lem1}} \label{prooflem1}
 \begin{proof}
  Note that $\mathbb{Q}_{k}\left(S_N=S_{k-H} u^i d^{\widetilde{N}+H-i} \lvert Z_{k,0}=1\right)$, $i=0,\dots,\widetilde{N}+H$ is the sum of several products of $\widetilde{N}$ elements chosen out of $\{\mathfrak{p}_{u},\mathfrak{p}_{d},\mathfrak{q}_{u},\mathfrak{q}_{d}\}$, and each product term corresponds to a path in the tree starting from the node $S_{k-H}$, and ending in the node $S_{N}=S_{k-H} u^i d^{\widetilde{N}+H-i}$. \par 
  
  Given equations (\ref{lastbigmove}), the last $H+1$ moves need to be either upward or downward, and they contribute to as just one single move. Since it is conditioned on $Z_{k,0}=1$, according to Remark (\ref{moves}), the first element in all of the product terms is either $\mathfrak{q}_{u}$ or $\mathfrak{p}_{u}$. For $H+1 \leq i \leq \widetilde{N}-2$, the last ($H+1$)-period move to $S_{N}=S_{k-H} u^i d^{\widetilde{N}+H-i}$ can be both downward and upward.\par 
  
   In the case that it is upward, we need to consider all the paths starting from $S_{k-H}$ to $S_{N-2}=S_{k-H} u^{i-2} d^{\widetilde{N}+H-i}$ which consist of $i-2$ upward moves and $\widetilde{N}+H-i$ downward ones. There are ${\widetilde{N}+H-2 \choose \i-2}$ of such paths, but these paths are not all equivalent and result in different product terms of $\widetilde{N}$ elements chosen out of $\{\mathfrak{p}_{u},\mathfrak{p}_{d},\mathfrak{q}_{u},\mathfrak{q}_{d}\}$, based on the location of the $\widetilde{N}+H-i$ downward moves in the path. \par 
   
   Note that all paths which have the same number of downward groups result in the same product terms, where a downward group is any number of consecutive downward moves preceded (if any) by 
  an upward move and also succeeded (if any) by an upward move. For example, both of the sequences $\nearrow \searrow \searrow \searrow \nearrow \searrow$ and $\searrow \searrow \nearrow \nearrow \searrow \searrow$ have two groups of $\searrow$ moves. The reason for studying downward groups is that the starting element in all of them is $\mathfrak{q}_{u}$. \par 
  
  In this notation, $j$ corresponds to the number of groups which starts from $1$ (assuming that there exists 
  at least one downward move) and can reach to $\min(i-H,\widetilde{N}+H-i)$. Notice that there are ${\widetilde{N}+H-i-1 \choose j-1}{i-H \choose j}$ paths which have exactly $j$ groups. Therefore, along 
  those path the power of both $\mathfrak{q}_{u}$ and $\mathfrak{p}_{d}$ is $j$ and consequently, the powers of $\mathfrak{q}_{d}$ and $\mathfrak{p}_{u}$ are respectively $\widetilde{N}+H-i-j$ and $i-j-H$. Here $f(i,j)$ in equation (\ref{fij}) corresponds to these paths. \par 
  
  The second case is that the last ($H+1$)-period move is downward. Then, we need to consider all the paths starting from the node $S_{k-H}$ to 
  $S_{N-2}=S_{k-H} u^i d^{\widetilde{N}+H-i-2}$ which consist of $i$ upward moves and $\widetilde{N}+H-i-2$ downward ones. Here not only the number of downward moves is important, but also the direction (upward or 
  downward) of the move from time $N-3$ to $N-2$ is also relevant. \par 
  
  Note that there are ${\widetilde{N}-i-2 \choose j-1}{i \choose j-1}$ paths which have exactly $j$ groups such that the last $1$-period
  move from $N-3$ to $N-2$ is downward, so the corresponding product term is $\mathfrak{q}_{u}^{(j)} \mathfrak{q}_{d}^{(\widetilde{N}-i-j)} \mathfrak{p}_{u}^{(i-j+1)} \mathfrak{p}_{d}^{(j-1)}$, and there are ${\widetilde{N}-i-2 \choose j-1}[{i+1 \choose j}-{i \choose j-1}]$ paths which have exactly $j$ groups such that the last $1$-period move from $N-3$ to $N-2$ is upward. The function $h(i,j)$ in equation (\ref{hij}) takes all these paths into account.\par 
  
  For $H+1 \leq i \leq \widetilde{N}-2$, it is necessary to use both $f(i,j)$ and $h(i,j)$ to take into account that the last ($H+1$)-period move can be both upward and downward. The same reasoning works for $0 \leq i \leq H$ and $\widetilde{N} \leq i \leq \widetilde{N}+H-1$, but here the last ($H+1$)-period move can only be downward for $0 \leq i \leq H$ and upward for $\widetilde{N} \leq i \leq \widetilde{N}+H-1$. For $i=\widetilde{N}-1$ when the last ($H+1$)-period
  move is downward and $i=\widetilde{N}+H$, the functions $f(i,j)$ and $h(i,j)$ cannot be used because in all of the paths from $S_{k-H}$ to $S_{N-2}=S_{k-H} u^{\widetilde{N}+H-2}$, there is not any downward move at all to
  make a downward group (i.e., $j=0$). 
  \end{proof}

\bibliographystyle{apalike}

\end{document}